\documentclass{llncs}

\usepackage{amsmath,amssymb,amsfonts}

\usepackage{graphicx}
\usepackage{appendix}
\usepackage{epsf}
\usepackage{epsfig}
\usepackage{epstopdf}
\usepackage{xspace}
\usepackage{soul}
\usepackage{latexsym}
\usepackage[linesnumbered,ruled,vlined]{algorithm2e}
\SetCommentSty{mycommfont}
\usepackage[linkbordercolor={0 0 1}]{hyperref}
\usepackage{nameref}

\usepackage{caption}
\usepackage{subcaption}
\usepackage[compatibility=false]{caption}
\newcommand{\footnoteref}[1]{\textsuperscript{\ref{#1}}}
\newcommand{\GI}{\mbox{{\sf GI}}}
\newcommand{\NP}{\mbox{{\sf NP}}}
\newcommand{\AC}{\mbox{{\sf AC}}}

\newcommand{\ISO}{\mbox{{\sf ISO}}}

\renewcommand{\P}{\mbox{{\sf P}}}
\newcommand{\fpt}{\mbox{{\sf fpt}}}

\title{Polynomial-time  Algorithm for Isomorphism of Graphs with Clique-width at most Three}
\author{Bireswar Das\thanks{Part of the research was done while the author was a DIMACS postdoctoral fellow.}, Murali Krishna Enduri\thanks{Supported by Tata 
Consultancy Services (TCS) research fellowship.} and I. Vinod Reddy}
\institute{IIT Gandhinagar, India \\
\email{\{bireswar,endurimuralikrishna,reddy\_vinod\}@iitgn.ac.in}
}

\begin{document}
\pagestyle{plain}

\date{}
\maketitle
\begin{abstract}
The clique-width is a measure of complexity of decomposing graphs
into certain tree-like structures. The class of graphs with bounded clique-width contains bounded tree-width graphs.
We give a polynomial time graph isomorphism algorithm for graphs with clique-width at most three.
Our work is independent of the work by
Grohe et al. \cite{grohe2015isomorphism} showing that the isomorphism problem for graphs of bounded clique-width is polynomial time. 
\end{abstract}
\section{Introduction}
\par Two graphs $G_1=(V_1,E_1)$ and $G_2=(V_2,E_2)$ are \emph{isomorphic} if there is a bijection 
$f:V_1\rightarrow V_2$ such that $\{u,v\}\in E_1$ 
if and only if $\{f(u),f(v)\}\in E_2$. Given a pair of graphs as input the problem of deciding if the two graphs are isomorphic is 
known as \emph{graph isomorphism problem} ($\GI$). 
Despite nearly five decades of research the complexity status of this problem still remains unknown. 
The graph isomorphism problem is not known to be in $\P$. It is in $\NP$ but very unlikely to be $\NP$-complete \cite{Boppana1987}. 
The problem is not even known to be hard for $\P$. 
Recently Babai \cite{babai2015graph} designed a quasi-polynomial time algorithm to solve the GI problem improving the previously 
best known $2^{O(\sqrt{n\log n})}$ time algorithm \cite{babai1981,Zem1982}.
Although the complexity of the general graph isomorphism problem remains elusive, many polynomial time algorithms are known for restricted classes of 
graphs e.g., bounded degree \cite{luks1982isomorphism}, bounded genus  \cite{miller1980isomorphism}, bounded tree-width \cite{bodlaender1990polynomial}, etc.

\par The graph parameter \emph{clique-width}, introduced by Courcelle et al. in \cite{courcelle1993handle}, has been studied extensively.  
The class of bounded clique-width graphs is  fairly large in the sense that it contains distance hereditary graphs, bounded 
tree-width graphs, bounded rank-width graphs \cite{kaminski2009recent}, etc.  Fellows et al. \cite{fellows2009clique} shows that the computing 
the clique-width of a graph is NP-hard. 
Oum and Seymour \cite{oum2006approximating} gave an elegant
algorithm that computes a $(2^{3k+2}-1)$-expression
for a graph $G$ of clique-width at most $k$ or decides that the clique-width is more than $k$.



The parameters tree-width and clique-width share some similarities, for example many $\NP$-complete problems admit polynomial 
time algorithms when the tree-width or the clique-width of the input graph is bounded. A polynomial time isomorphism algorithm for
bounded tree-width graphs has been known for a long time \cite{bodlaender1990polynomial}. Recently Lokhstanov et al. \cite{FPTTW2014} gave an $\fpt$ algorithm for 
$\GI$ parameterized by tree-width. 
The scenario is different for bounded clique-width graphs. The complexity of $\GI$  for bounded clique-width graphs is not known. Polynomial time 
algorithm for $\GI$ for graphs with clique-width 
at most $2$, which coincides with the class of co-graphs, is known probably as a folklore.
The complexity of recognizing graphs with clique-width at most three was unknown until Corneil et al. \cite{corneil2012polynomial} 
came up with the first polynomial time algorithm. 
Their algorithm (henceforth called the CHLRR algorithm) works via an extensive study of the structure of such graphs using split and modular decompositions. 
Apart from recognition, the CHLRR algorithm also produces a $3$-expression for graphs with clique-width at most three. For fixed $k>3$, though algorithms 
to recognize graphs with clique-width at most $k$ are known \cite{oum2006approximating}, computing a $k$-expression is still open. 
Recently in an independent work by Grohe et al. \cite{grohe2015isomorphism} designed an isomorphism algorithm for graphs of bounded 
clique-width subsuming our result.  Their algorithm uses group theory techniques and has worse runtime.
However our algorithm has better runtime and uses different simpler intuitive techniques. 

In this paper we give isomorphism algorithm for graphs with clique-width at most three with runtime $O(n^3m)$. 
Our algorithm works via first defining a notion of equivalent $k$-expression and designing $O(n^3)$ algorithm to test if two input $k$-expressions 
are equivalent under this notion.
Next we modify the CHLRR algorithm slightly to output a linear sized set $parseG$ of $4$-expressions for an input graph 
$G$ of clique-width at most three which runs in $O(n^3m)$ time. 
Note that modified CHLRR algorithm will not output a canonical expression. However   
we show that for two isomorphic graphs $G$ and $H$ of clique-width at most three, $parseG$ contains an equivalent 
$k$-expression for each $k$-expression in $parseH$ and vice versa. Moreover, if $G$ and $H$ are not isomorphic then no pair in $parseG\times parseH$  is equivalent.

\section{Preliminaries}
In this paper, the graphs we consider are without multiple edges and self loops. The complement of a graph $G$ is denoted as
$\overline G$.
The {\it coconnected components} of $G$ are the connected 
components of $\overline{G}$. We say that a vertex $v$ is {\it universal} to a vertex set $X$ if $v$ is adjacent to all vertices in $X \setminus \{v\}$. 
A {\it biclique} is a bipartite graph $(G,X,Y)$, such that every vertex in $X$ is connected to every 
vertex of $Y$.
A {\it labeled graph} is a graph with labels assigned to vertices such that each vertex has exactly one label. 
In a labeled graph $G$, $lab(v)$ is the label of a vertex $v$ and $lab(G)$ is the set of all labels. 
We say that a graph is {\it bilabeled} ({\it trilabeled}) if it is labeled using exactly two (three) labels.  
The set of all edges between vertices of label $a$ and label $b$ is denoted $E_{ab}$. We say $E_{ab}$ is complete
if it corresponds to a biclique.

The subgraph of $G$ induced by $X \subseteq V(G)$ is denoted by $G[X]$,  the set of vertices adjacent to $v$ is denoted $N_G(v)$. 
The closed neighborhood $N_G[v]$ of $v$ is $N_G(v) \cup \{v\}$. We write $G \cong _f H$ if $f$ is an isomorphism between graphs $G$ and $H$. 
For labeled graphs $G$ and $H$, we write  $G \cong _f^{\pi} H$ if $G \cong _f H$ and $\pi:lab(G)\rightarrow lab(H)$ is a bijection 
such that for all $x \in V(G)$ if $lab(x)=i$ then $lab(f(x))=\pi(i)$. The set of all isomorphisms from $G$ to $H$  is denoted $\ISO(G,H)$.

\begin{definition}
 The {\bf clique-width} of a graph $G$ is defined as the minimum number of labels needed to construct $G$ 
 using the following four operations:
 \begin{enumerate}
 \setlength{\itemsep}{1pt}
\setlength{\parskip}{0pt}
  \item [i.] {\bf $v(i)$}: Creates a new vertex $v$ with label $i$
  \item [ii.] {\bf $G_1\oplus G_2 \cdots \oplus G_l$}: Disjoint union of labeled graphs $G_1,G_2,\cdots,G_l$
  \item [iii.] {\bf $\eta_{i,j}$}: Joins each vertex with label $i$ to each vertex with label $j$ ($i \neq j$)
  \item [iv.] {\bf $\rho_{i\rightarrow j}$}: Renames all vertices of label $i$ with label $j$ 
 \end{enumerate}
\end{definition}
Every graph can be constructed using the above four operations, which is represented by an  
algebraic expression known as $k$-$expression$, where $k$ is the number of labels used in expression.
The {\it clique-width} of a graph $G$, denoted by $cwd(G)$, is the minimum $k$ for which there exists a $k$-expression that 
defines the graph $G$. 
From the $k$-expression of a graph we can construct a tree known as {\it parse tree} of $G$.
The leaves of the parse tree are vertices of $G$ with their initial labels, and the internal nodes 
correspond to the operations ($\eta_{i,j}$, $\rho_{i \rightarrow j}$ and $\oplus$) used to construct $G$.
For example, $C_5$ (cycle of length 5) can be constructed by
$$ \eta_{1,3}((\rho_{3\rightarrow 2}(\eta_{2,3}((\eta_{1,2}(a(1)\oplus b(2)))\oplus (\eta_{1,3}(c(3)\oplus d(1))))))\oplus e(3)).$$
The $k$-expression for a graph need not be unique.
The clique-width of any induced subgraph is at most the clique-width of its graph \cite{courcelle2000upper}.

\par Now we describe the notions of modular and split decompositions.
A set $M \subseteq V(G)$ is called a {\it module} of $G$ if all vertices of $M$ have the same set 
of neighbors in $V(G)\setminus M$.  
The \emph{trivial modules} are $V(G)$, and $\{v\}$ for all $v$.
In a labeled graph, a module is said to be a {\it $l$-module} if all the vertices in the module have the same label.
A {\it prime} ({\it $l$-prime}) graph  is a graph (labeled graph) in which all modules ($l$-modules) are trivial. 
The modular decomposition of a graph is one of the decomposition techniques which was introduced by Gallai \cite{gallai1967transitiv}.
The {\it modular decomposition} of a graph $G$ is a rooted tree $T^G_M$ 
that has the following properties:
\begin{enumerate}
\setlength{\itemsep}{1pt}
\setlength{\parskip}{0pt}
 \item The leaves of $T^G_M$ are the vertices of $G$.
 \item For an internal node $h$ of $T^G_M$, let $M(h)$ be the set of vertices of $G$ that are leaves of the subtree of 
 $T^G_M$ rooted at $h$. ($M(h)$ forms a module in $G$).
 \item For each internal node $h$ of $T^G_M$ there is a graph $G_h$ (\emph{representative graph}) with 
 $V(G_h)=\{h_1,h_2,\cdots,h_r\}$, where $h_1,h_2,\cdots,h_r$ are the 
 children of $h$ in $T^G_M$ and for  $1 \leq i<j \leq r$, $h_i$ and $h_j$ are adjacent in $G_h$ iff there are vertices $u \in M(h_i)$ and $v \in M(h_j)$ that are adjacent in $G$.
 \item $G_h$ is either a clique, an independent set, or a prime graph and 
 $h$ is labeled \emph{Series} if $G_h$ is clique, \emph{Parallel} if $G_h$ is an independent set, and \emph{Prime} otherwise.
\end{enumerate}

 James et al. \cite{james1972graph} gave first polynomial time algorithm for finding a
 modular decomposition which runs in $O(n^4)$ time. Linear time algorithms to find modular decompositions are proposed 
in \cite{cournier1994new,tedder2008simpler}.

A vertex partition $(A,B)$ of a graph $G$ is a {\it split}
if $\tilde{A}=A\cap N(B)$ and $\tilde{B}=B\cap N(A)$ forms a biclique.
A split is trivial if $|A|$ or $|B|$ is one.
Split decomposition was introduced by Cunningham \cite{william1980combinatorial}. 
Loosely it is the result of a recursive process of decomposing a graph into components based on the splits.
Cunningham \cite{william1980combinatorial} showed that a graph can be decomposed uniquely into components that are stars, cliques, or prime (i.e., 
without proper splits). This decomposition is known as the \emph{skeleton}. For details see \cite{cunningham1982decomposition}.
A polynomial time algorithm for computing the skeleton of a graph is given in \cite{ma1994n}.
%
\begin{theorem}{\cite{cunningham1982decomposition}}({see \cite{corneil2012polynomial}})\label{th1}
 Let $G$ be a connected graph. Then the skeleton of $G$ is unique, and the proper splits of $G$ correspond to the special edges 
 of its skeleton and to the proper splits of its complete and star components. 
\end{theorem}
{\bf Organization of the paper: }
In Section~\ref{sec3} we discuss $\GI$-completeness of prime graph isomorphism.
In Section~\ref{sec4} we define a notion of equivalence of parse trees called \emph{structural isomorphism}, and give an algorithm to test 
if two parse trees are structurally isomorphic.
We give an overview of the CHLRR algorithm \cite{corneil2012polynomial} in Section~\ref{sec5}.
In Section~\ref{sec6}, we present the isomorphism algorithm for prime graphs of clique-width at most three. 
In Appendix, we show that the CHLRR algorithm can be modified suitably to output structurally isomorphic parse trees for isomorphic graphs.


\section{Completeness of Prime Graph Isomorphism}\label{sec3}
It is known that isomorphism problem for prime graphs is $\GI$-complete \cite{Bonamy2010}. There is an easy polynomial time many-one reduction from 
$\GI$ to prime graph isomorphism\footnote{In fact, it is an $\AC^0$ reduction} described in Lemma~\ref{PrimeGI} of the Appendix.
Unfortunately, this reduction does not preserve the clique-width. 
We also give a clique-width preserving Turing reduction from $\GI$ to prime graph isomorphism which we use in our main algorithm.
The reduction hinges on the following lemma. 
\begin{lemma}\label{prime} \cite{courcelle2000linear}
 $G$ is a graph of clique-width at most $k$ iff each prime graph associated with the modular decomposition of $G$ is of 
 clique-width at most $k$.
\end{lemma}
We next show that if we have an oracle for $\GI$ for colored prime graphs of clique-width at most $k$ then there is a $\GI$ algorithm for 
graphs with clique-width at most $k$. 
\begin{theorem}\label{primec}
 Let $\mathcal{A}'$ be an algorithm that given two colored prime graphs $G'$ and $H'$ of clique-width at most $k$, decides 
 if $G' \cong H'$ via a color preserving isomorphism. Then there exists an algorithm $\mathcal{A}$ that on input any colored graphs $G$ and $H$
 of clique-width at most $k$ decides if $G \cong H$ via a color preserving isomorphism.
\end{theorem}
\begin{proof}
 Let $G$ and $H$ be two colored graphs of clique-width at most $k$. The algorithm is similar to \cite{das2015logspace}, 
 which proceeds in a bottom up approach in stages starting from the leaves to 
 the root of the modular decomposition trees $T_G$ and $T_H$ of $G$ and $H$ respectively.
 Each stage corresponds to a level in the modular decomposition. 
 In every level, the algorithm $\mathcal{A}$ maintains a table that stores whether for each pair of nodes $x$ and $y$ in $T_G$ and $T_H$ the 
 subgraphs $G[x]$ and $H[y]$ induced by leaves of subtrees of $T_G$ and $T_H$ rooted at $x$ and $y$ are isomorphic.
  For the leaves it is trivial to store such information. Let $u$ and $v$ be two internal nodes in 
 the modular decomposition trees of $T_G$ and $T_H$ in the same level. 
 To decide if $G[u]$ and $H[v]$ are isomorphic $ \mathcal{A}$ does the following. 
\par If $u$ and $v$ are both {\it series} nodes then it just checks if the
 children of $u$ and $v$ can be isomorphically matched. The case for {\it parallel} node is similar. 
 If $u$ and $v$ are {\it prime} nodes then the vertices of representative graphs $G_u$ and $H_v$ are colored by 
 their isomorphism type i.e., two internal vertices $u_1$ and $u_2$ of the representative graphs
 will get the same color iff subgraphs induced by leaves of subtrees of $T_ G$ (or $T_H$) rooted at $u_1$ and $u_2$ are isomorphic.
 To test $G[u] \cong H[v]$, $ \mathcal{A}$ calls $\mathcal{A}'(\widehat G_u,\widehat H_v)$, where $\widehat G_u$ and $\widehat H_v$ 
 are the colored copies of $G_u$ and $H_v$ respectively. At any level if we can not  find a pairwise isomorphism matching between the internal nodes in that level 
 of $T_G$ and $T_H$ then $G \cong H$. In this manner we make $O(n^2)$ calls to algorithm $\mathcal{A}'$ at each level. The total runtime 
 of the algorithm is $O(n^3) T(n)$, where $T(n)$ is run time of $\mathcal{A}'$. 
 Note that by Lemma \ref{prime} clique-width of $G_u$ and $H_v$ are at most $k$. \qed
\end{proof}
\section{Testing Isomorphism between Parse Trees}\label{sec4}
In this section we define a notion of equivalence of parse trees called \emph{structural isomorphism}, and we give an algorithm to test 
if two given parse trees are equivalent under this notion. As we will see, the graphs generated by equivalent parse trees are always isomorphic. 
Thus, if we have two equivalent parse trees for the two input graphs, 
the isomorphism problem indeed admits a polynomial time algorithm. In Section~\ref{sec6}, we prove that the CHLRR algorithm can be tweaked 
slightly to produce structurally isomorphic parse trees for isomorphic graphs with clique-width at most three and thus giving a polynomial-time 
algorithm for such graphs.

Let $G$ and $H$ be two colored graphs. A bijective map $\pi:V(G)\rightarrow V(H)$ is \emph{color consistent} if
for all vertices $u$ and $v$ of $G$, $color(u)=color(v)$ iff $color(\pi(u))=color(\pi(v))$.
Let $\pi:V(G)\rightarrow V(H)$ be a color consistent mapping, define $\pi/color:color(G)\rightarrow color(H)$ as follows: for all 
$c$ in $color(G)$, $\pi/color(c)=color(\pi(v))$ where $color(v)=c$. It is not hard to see that the map $\pi/color$ is well defined.
Recall that the internal nodes of a parse tree are $\eta_{i,j}$, $\rho_{i \rightarrow j}$ and $\oplus$ operations.
The levels of a parse tree correspond to $\oplus$ nodes.  
Let $T_g$ be a parse tree of $G$  rooted at $\oplus$ node $g$. Let $g_1$ be descendant of $g$ which is neither $\eta$ nor $\rho$. 
We say that $g_1$ is an \emph{immediate significant descendant} of $g$ if there is no other $\oplus$ node in the path from $g$ to  $g_1$.
For an immediate  significant descendant $g_1$ of $g$, we construct a \emph{colored quotient graph} $Q_{g_1}$ that corresponds to 
graph operations appearing in the path from $g$ to $g_1$ performed on graph $G_{g_1}$, where $G_{g_1}$ is graph generated by parse tree $T_{g_1}$.
The vertices of $Q_{g_1}$ are labels of $G_{g_1}$. The colors and the edges of $Q_{g_1}$ are determined by 
the operations on the path from $g_1$ to $g$. We start with coloring a vertex $a$ by color $a$ and no edges. If the operation performed is $\eta_{a,b}$ on $G_{g_1}$ then 
add edges between vertices of color $a$ and color $b$. If the operation is $\rho_{a\rightarrow b}$ on $G_{g_1}$ then recolor 
the vertices of color $a$ with color $b$. After taking care of an operation we move to the next operation on the path 
from $g_1$ to $g$ until we reach $\oplus$ node $g$. Notice that if the total number of labels used in a parse tree is $k$ then the 
size of any colored quotient graph is at most $k$.

\begin{definition}\label{structiso}
 Let $T_g$ and $T_h$ be two parse trees of $G$ and $H$ rooted at $\oplus$ nodes $g$ and $h$ respectively.
 We say that $T_g$ and $T_h$ are \emph{structurally isomorphic via a label map $\pi$} (denoted $T_g\cong^{\pi}T_h$)    
  \begin{enumerate}
  \setlength{\itemsep}{1pt}
\setlength{\parskip}{0pt}
 \item If $T_g$ and $T_h$ are single nodes\footnote{In this case they are trivially {structurally isomorphic} via $\pi$.} or inductively,
 \item If $T_g$ and $T_h$ are rooted at $g$ and $h$ having immediate  significant descendants $g_1,\cdots,g_r$ and $h_1,\cdots,h_r$, and 
  there is a bijection $\gamma : [r] \rightarrow [r]$  and for each $i$ there is a $\pi_i$ $\in$ $\ISO(Q_{g_i},Q_{h_{\gamma(i)}})$ such that 
 $T_{g_i} \cong^{\pi_i} T_{h_{\gamma(i)}}$ and $\pi_i/color=\pi|_{color(Q_{g_i})}$, where $T_{g_1}, \cdots ,T_{g_r}$ and $T_{h_1}, \cdots ,T_{h_r}$ are 
the subtrees rooted at $g_1,\cdots,g_r$ and $h_1,\cdots,h_r$ respectively\footnote{Notice that this definition 
implies that $G_{g_i}$ and $H_{h_{\gamma(i)}}$ are isomorphic via the label map $\pi_{i}$ where 
$G_{g_i}$ and $H_{h_{\gamma(i)}}$ are graphs generated by the parse trees $T_{g_i}$ and $T_{h_{\gamma(i)}}$ respectively.}
 \end{enumerate}
 We say that $T_g$ and $T_h$ are \emph{structurally isomorphic} if there is a $\pi$ such that $T_g\cong^\pi T_h$.
 \end{definition}
The structural isomorphism is an equivalence relation: reflexive and symmetric properties are immediate from the above definition.
 The following lemma shows that it is also transitive.
\begin{lemma}\label{deftrans}
Let $T_{g_1}$, $T_{g_2}$ and $T_{g_3}$ be the parse trees of $G_1$, $G_2$ and $G_3$ respectively such that 
$T_{g_1}\cong^{\pi_1}T_{g_2}$ and $T_{g_2}\cong^{\pi_2}T_{g_3}$ then 
$T_{g_1}\cong^{\pi_2\pi_1}T_{g_3}$.
\end{lemma}
\begin{proof} The proof is by induction on the height of the parse trees.
 The base case trivially satisfies the transitive property. Assume that $g_1$, $g_2$ and $g_3$ are nodes of height $d+1$. 
 Let $g_{1i}$ be an immediate  significant descendant of $g_1$. Since $T_{g_1}\cong^{\pi_1} T_{g_2}$, there is an immediate significant 
 descendant $g_{2j}$ of $g_2$ and $\pi_{1i}\in \ISO(Q_{g_{1i}},Q_{g_{2j}})$ such that  $\pi_{1i}/color=\pi|_{color(Q_{g_{1i}})}$ and 
 $T_{g_{1i}} \cong ^{\pi_{1i}} T_{g_{2j}}$. Similarly, $g_{2j}$ will be matched to some immediate significant descendant $g_{3k}$ of 
 $g_3$ via $\pi_{2j}\in \ISO(Q_{g_{2j}},Q_{g_{3k}})$ such that  $\pi_{2j}/color=\pi|_{color(Q_{g_{2j}})}$ and $T_{g_{2j}} \cong ^{\pi_{2j}} T_{g_{3k}}$. 
 The nodes $g_{1i}, g_{2j}$ and $g_{3k}$ has height at most $d$. Therefore, by induction hypothesis $T_{g_{1i}} \cong ^{\pi_{2j}\pi_{1i}} T_{g_{3k}}$. 
 By transitivity of isomorphism we can say $\pi_{2j}\pi_{1i}\in \ISO(Q_{g_{1i}},Q_{g_{3k}})$.
To complete the proof we just need to show $\pi_{2j}\pi_{1i}/color = \pi_{2}\pi_{1}|_{color(Q_{g_{1i}})}$. This can be inferred from the following two facts:\\ 
   1) $\pi_{2j}\pi_{1i}/color=\pi_{2j}/color$  $\pi_{1i}/color$\\
   2) $\pi_{2}\pi_{1}|_{color(Q_{g_{1i}})}=\pi_{2}|_{color(Q_{g_{2j}})}$  $\pi_{1}|_{color(Q_{g_{1i}})}$.   \qed

\end{proof}
 
\noindent
 {\bf Algorithm to Test Structural Isomorphism: } Next we describe an algorithm that given two parse trees $T_G$ and $T_H$ tests 
 if they are structurally isomorphic. From the definition 
if $T_G \cong^{\pi} T_H$  then we can conclude that $G$ and $H$ are isomorphic. 
We design a dynamic programming algorithm that basically checks the local conditions 1 and 2 in Definition~\ref{structiso}.
%

The algorithm starts from the leaves of parse trees and proceeds in levels where each level corresponds to 
$\oplus$ operations of parse trees. Let $g$ and $h$ denotes the $\oplus$ nodes at level $l$ of $T_G$ and $T_H$ respectively.
At each level $l$, for each pair of $\oplus$ nodes $(g,h) \in (T_G,T_H)$, the algorithm computes the set $R_l^{g,h}$ of all bijections 
$\pi: lab(G_{g}) \rightarrow lab(H_{h})$ such that $G_{g} \cong_f^{\pi} H_{h}$ for some $f$, and stores in a table indexed by $(l,g,h)$, 
where $G_{g}$ and $H_{h}$ are  graphs generated by sub parse trees $T_g$ and $T_h$ rooted at $g$ and $h$ respectively.
 To compute $R_l^{g,h}$, the algorithm uses the already computed information $R_{l+1}^{g_i,h_j}$ where $g_i$ and $h_j$ are immediate significant descendants of $g$ and $h$.
  \par The base case correspond to finding $R_l^{g,h}$ for all pairs $(g,h)$ such that $g$ and $h$ are leaves. 
Since in this case $G_{g}$ and $H_{h}$ are just single vertices, it is easy 
to find $R^{g,h}_l$. 
For the inductive step let $g_1,\cdots,g_r$ and $h_1,\cdots,h_{r'}$ be the immediate significant descendants of $g$ and $h$ respectively.  
If $r \neq r'$ then $R^{g,h}_l=\emptyset$. 
Otherwise we compute $R^{g,h}_l$ for each pair $(g,h)$ at level $l$ with help of the already computed information up to level $l+1$ as follows.
\par For each $\pi: lab(G_g) \rightarrow lab(H_h)$ and 
pick $g_1$ and try to find a $h_{i_1}$ such that $T_{g_1} \cong^{\pi_1} T_{h_{i_1}}$ for some $\pi_1\in 
\ISO(Q_{g_1},Q_{h_{i_1}})\cap R^{g_1,h_{i_1}}_{l+1}$ such that $\pi_1/color=\pi|_{color(Q_{g_1})}$.
We do this process to pair $g_2$ with some unmatched $h_{i_2}$.
Continue in this way until all immediate significant descendants are matched. By Lemma~\ref{lemmatrans}, we know that this greedy matching 
satisfies the conditions of Definition~\ref{structiso}.
If all the immediate significant descendants are matched we add $\pi$ to $R^{g,h}_l$. 
It is easy to see that if $R^{g,h}_l \neq \emptyset$ then the subgraphs $G_g \cong_f^{\pi} H_h$ for $\pi \in R^{g,h}_l$.  
From the definition of structurally isomorphic parse trees it is clear that if $R^{g,h}_0 \neq \emptyset$ then 
$G \cong H$. The algorithm is polynomial time as the number of choices for $\pi$ and $\pi_1$ is at most $k!$ which is a constant, where $|lab(G)|=k$.

Note that for colored graphs, by ensuring that we only match vertices of same color in the base case, the whole algorithm can be made to work for 
colored graphs. In Lemma~\ref{deftrans} we prove that structural isomorphism satisfies transitivity. In fact, structural isomorphism satisfies a stronger notion of 
transitivity as stated in the following lemma.
\begin{lemma}\label{lemmatrans}
Let $T_g$ and $T_h$ be two parse trees of graphs $G$ and $H$. Let $g_1$ and $g_2$ be two immediate significant descendants of 
$g$, and $h_1$ and $h_2$ be two immediate significant descendants of $h$. Suppose for $i=1,2$, $T_{g_i}\cong^{\pi_i} T_{h_i}$ for some $\pi_i\in \ISO(Q_{g_i},Q_{h_i})$ 
with $\pi_i/color=\pi|_{color(Q_{g_i})}$. Also assume that $T_{g_1}\cong^{\pi_3} T_{h_2}$ where $\pi_3\in \ISO(Q_{g_1},Q_{h_2})$ 
and $\pi_3/color=\pi|_{color(Q_{g_1})}$. Then, $T_{g_2}\cong^{\pi_1\pi^{-1}_3\pi_2} T_{h_1}$ where $\pi_1\pi^{-1}_3\pi_2\in \ISO(Q_{g_2},Q_{h_1})$ 
and $\pi_1\pi^{-1}_3\pi_2/color=\pi|_{color(Q_{g_2})}$.
\end{lemma}
\begin{proof}
  By Lemma~\ref{deftrans},  $T_{g_2}\cong^{\pi_1\pi^{-1}_3\pi_2} T_{h_1}$. The rest of the proof is similar to the proof of the inductive case of Lemma~\ref{deftrans}. \qed
\end{proof}

\section{Overview of the CHLRR Algorithm}\label{sec5}
 Corneil et al. \cite{corneil2012polynomial} gave the first polynomial time algorithm (the \emph{CHLRR algorithm}), 
 to recognize graphs of clique-width at most three. We give a brief description of their algorithm in this section. 
 We mention that our description of this fairly involved algorithm is far from being complete. The reader is encouraged to 
 see \cite{corneil2012polynomial} for details. By Lemma~\ref{prime} we assume that the input graph $G$ is prime. 

To test whether clique-width of prime graph $G$ is at most three the algorithm starts by  
constructing a set of bilabelings and trilabelings of $G$. In general the number of bilabelings and trilabelings are 
exponential, but it was shown (Lemma 8 and 9 in \cite{corneil2012polynomial} ) that it is enough to consider the following	
linear size subset denoted by $LabG$. 
\begin{itemize}
\item [1.] For each vertex $v$ in $V(G)$
\begin{itemize}
\item [[$B_1$ \hspace{-0.3cm}]] Generate the bilabeling\footnote{ {\it bilabeling} of a set $X \subseteq V$ indicates that all the vertices in $X$ are 
labeled with one label and $V\setminus X$ is labeled with another label. \label{note1} }$\{v\}$ and add it to $LabG$.
\item [[$B_2$\hspace{-0.2cm}]] Generate the bilabeling $\{x \in N(v) \hspace{0.1cm} | \hspace{0.1cm} N[x] \subseteq N[v] \}$ and add it to $LabG$.
\end{itemize}

\item [2.] Compute the skeleton of $G$ search this skeleton for the special edges, clique and star components.
\begin{itemize}
\item [[$T_1$\hspace{-0.2cm}]] For each special edge $s$ (corresponds to a proper split), generate the trilabeling $\tilde{X},\tilde{Y}, V(G)\setminus (\tilde{X} \cup \tilde{Y})$ where ($X$,$Y$) is the split defined by $s$ and add it to $LabG$.
\item [[$B_3$\hspace{-0.2cm}]] For all clique components $C$, generate the bilabeling $C$ and add it to $LabG$.
\item [[$B_4$\hspace{-0.2cm}]] For all star components $S$,  generate the bilabeling $\{c\}$, where $c$ is the special center of $S$, and add it to $LabG$. 
\end{itemize}

\end{itemize}
\begin{lemma} \label{over} \cite{corneil2012polynomial}
Let $G$ be a prime graph. Clique-width of $G$ is at most three if and only if at least one of the bilabelings or
trilabelings in $LabG$ has clique-width at most three.
\end{lemma}
By Lemma~\ref{over} the problem of testing whether $G$ is of clique-width at most three is reduced to checking 
one of labeled graph in $LabG$ is of clique-width at most three. 
To test if a labeled graph $A$ taken from $LabG$ is of clique-width at most three, the algorithm follows a top down approach by iterating over all possible last operations that arise in the parse tree representation of $G$. 
For example, for each vertex $x$ in $G$ the algorithm checks whether the last operation must have joined $x$ with 
its neighborhood. In this case the problem of testing whether $G$ can be constructed using at most three labels is reduced to test whether $G \setminus \{x\}$ can be constructed using at most three lables. 
Once the last operations are 
fixed the original graph decomposes into smaller components, which can be further decomposed recursively.  

For each $A$ in $LabG$, depending on whether it is bilabeled or trilabeled the algorithm makes different tests on $A$ to 
determine whether $A$ is of clique-width at most three.
Based on the test results the algorithm either concludes clique-width of $A$ is more than three or returns top operations of the parse tree for $A$ along with some connected components of $A$ which are further decomposed recursively.

If $A$ in $LabG$ is connected, trilabeled (with labels $l_1, l_2, l_3$) and $l$-prime then by the
construction of $LabG$, $A$ corresponds to a split (possibly trivial).
If $A$ has a proper split then there exists $a \neq b$ in $\{l_1, l_2, l_3\}$ such that $A$ will be disconnected with the removal of edges $E_{ab}$. This 
gives a decomposition with top operations $ \eta _{a,b}$ followed by a $\oplus$ node whose children are connected components of $A\setminus E_{ab}$. If $A$ has a universal vertex $v$ (trivial split) labeled $a$ in $A$ then by
removing edges $E_{ab}$ and $E_{ac}$ we get a decomposition with top operations $\eta_{a,b}$ and $\eta_{a,c}$ followed by a $\oplus$ operation with children connected components of $A\setminus (E_{ab} \cup E_{ac})$ .

To describe the bilabeled case we use $V_i$ to denote the set of vertices of $A$ with label $i$. If $A$ in $LabG$ is 
connected, bilabeled (with labels $l_1, l_2$) and $l$-prime, 
then the last operation is neither $\eta_{l_1,l_2}$ (otherwise $A$ will have a $l$-module) nor $\oplus$ ($A$ is connected). 
So the last operation of the decomposition must be a relabeling followed by a join operation i.e., we have to introduce a 
third label set $V_{l_3}$ such that all the edges are present between the two of three labeled sets. 
 
 After introducing third label if there is only one join to undo, then we have a unique way to decompose the graph 
into smaller components. If there are more than one possible join to be removed, then it is enough to consider one of them and proceed 
(see Section 5.2 in \cite{corneil2012polynomial}).  
There are four ways to introduce the third label to decompose the graph, but they might correspond to overlapping cases.
To overcome this the algorithm first
checks whether $A$ belongs to any of 
three simpler cases described below. 


{\bf PC1:} $A$ has a universal vertex $x$ of label $l \in \{l_1,l_2\}$. In this case relabel vertex $x$ with $l_3$ and remove the edges $E_{l_3l_2}$, and 
$E_{l_3l_1}$ to decompose $A$. This gives a decomposition with $\rho_{l_3\rightarrow l}$, $\eta_{l_3,l_2}$, $\eta_{l_3,l_1}$ followed by $\oplus$ operation with children $x$ and $A \setminus \{x\}$.

{\bf PC2:} $A$ has a vertex $x$ of label $l\in\{l_1,l_2\}$ that is universal to all vertices of label $l' \in \{l_1,l_2\}$, but is not adjacent to all 
vertices with the other label, say $\bar{l'}$. In this case relabel vertex $x$ with $l_3$ and remove the edges $E_{l_3l'}$. This gives a decomposition with $\rho_{l_3\rightarrow l}$, $\eta_{l_3,l'}$ above a $\oplus$ operation with children $x$ and $A \setminus \{x\}$.

{\bf PC3:} $A$ has two vertices $x$ and $y$ of label $l$, where $y$ is universal to everything other than $x$, 
and $x$ is universal to all vertices of label $l$ other than $y$, and non-adjacent to all vertices with the other label $\bar l$.
In this case the algorithm relabels vertices $x$ and $y$  with $l_3$, and by removing edges $E_{l_3l}$ disconnects the graph $A$, 
with two connected components $x$ and $A\setminus \{x\}$. Now in graph $A\setminus \{x\}$ again remove the edges $E_{l_3 \bar l}$ to decompose the graph into two parts $y$ and $A\setminus \{x,y\}$.

\par If $A$ does not belongs to any of above three simpler cases then there are 
four different ways 
to introduce the third label set to decompose the graph as described below. 
 
\par 
Let $\mathcal{E}$ be the set of all connected, bilabeled, $l$-prime graphs with clique-width at most three and not 
belonging to above three simpler cases.
For $l \in \{1,2\}$ we define the following four subsets of $\mathcal{E}$. 
\begin{itemize}
\item [1.] $\mathcal{U}_{l}$: $V_l^a \neq \emptyset$ and removing the edges between the $V_l^a$ and $V_{\bar l}$ disconnects
the graph.
\vspace{0.1cm}
\item [2.] $\mathcal{\overline{D}}_{l}$: $\overline V_{l}$ is not connected and removing the edges between the coconnected 
components of $\overline V_{l}$ disconnects the graph. 
\end{itemize}
In these four cases the algorithm introduces a new label $l_3$ and removes the edges $E_{ll_3}$, $l \in\{l_1,l_2\}$ to disconnect $A$. 
This gives a decomposition with $\rho_{l_3 \rightarrow l}$ and $\eta_{l,l_3}$ followed by $\oplus$ operation with children 
that are the connected components of $A\setminus E_{ll_3}$. For more details about decomposition process when $A$ is in $\mathcal{U}_{l}$ or
$\mathcal{\overline{D}}_{l}$, $l \in \{1,2\}$ 
the reader is encouraged to see Section 5.2 in \cite{corneil2012polynomial}.
\par The following Lemma shows that there is no other possible way of decomposing a clique-width 
at most three graphs apart from the cases described above. 
\vspace{-0.5cm}
\begin{lemma} \cite{corneil2012polynomial}
$\mathcal {E}= \mathcal{U}_{1} \cup \mathcal{U}_{2} \cup \mathcal{\overline{D}}_{1} \cup \mathcal{\overline{D}}_{2}$, 
 and this union is disjoint. 
\end{lemma}
In summary, for any labeled graph $A$ in $LabG$ the CHLRR algorithm tests whether $A$ belongs to any of the above described cases, if
it is then it outputs suitable top operations and connected components. The algorithm continues the above process repeatedly on 
each connected component of $A$ until it either returns a parse tree or concludes clique-width of $A$ is more than three.
\section{Isomorphism Algorithm for Prime Graphs of Clique-width at most Three}\label{sec6}
In Section~\ref{sec4} we described algorithm to test structural isomporphism between two parse trees. 
In this Section we show that 
given two isomorphic prime graphs $G$ and $H$ of clique-width at most three, the CHLRR algorithm can be slightly modified 
to get structurally isomorphic parse trees. We have used four labels in order to preserve structural isomorphism in the modified algorithm.
The modified algorithm is presented in Appendix.
Recall that the first step of the CHLRR algorithm is to construct a set $LabG$ of bilabelings and trilabelings of $G$ as described in 
Section~\ref{sec5}.
\begin{definition}
We say that $LabG$ is {\it equivalent} to $LabH$ denoted as $LabG\equiv LabH$ if there is a bijection $g: LabG \rightarrow LabH$ such that for all $A \in LabG$, there is  
an isomorphism $f : V(A) \rightarrow V(g(A))$ and a bijection $\pi : lab(A) \rightarrow lab(g(A))$ 
such that $A \cong^{\pi}_f g(A)$.
\end{definition}

\begin{lemma}\label{th12}
$LabG\equiv LabH$ iff $G \cong_f H$.
\end{lemma}
\begin{proof}
The proof follows from the construction of sets $LabG$ and $LabH$ from input prime graphs $G$ and $H$ and it is presented in Appendix. \qed 
\end{proof}

\begin{lemma}\label{de}
 Let $A \in LabG$ and $B\in LabH$. If $A \cong_f^{\pi} B$ for some $f$ and $\pi$ then parse trees generated from $Decompose$ 
function (Algorithm~\ref{algodecompose}) for input graphs $A$ and $B$ are structurally isomorphic. 
More specifically, $Decompose(A) \cong_f^{\pi} Decompose(B)$. 
\end{lemma}
\begin{proof}
 Follows from Lemma~\ref{TI} and Lemma~\ref{BI} described in Appendix. The major modifications are done in PC2 case, where we have used 
 four labels in order to preserve structural isomorphism between parse trees. \qed
\end{proof}



{\bf Isomorphism Algorithm} \\
For two input prime graphs $G$ and $H$ the algorithm works as follows. Using modified CHLRR algorithm, first a parse tree $T_G$ of clique-width at most three is computed for $G$. The parse tree $T_G$ of $G$ is not canonical but from Lemma~\ref{th12} and~\ref{de}, we know that if $G \cong H$ then there exists parse tree $T_H$ of $H$, structurally isomorphic to $T_G$.  
Therefore we compute parse tree of clique-width at most three for each labeled graph in $LabH$. For each such parse tree $T_H$, the algorithm uses
the structural isomorphic algorithm described in Section~\ref{sec4} to test the structural isomorphism between  parse trees $T_G$ and $T_H$.
If $T_G \cong T_H$ for some $T_H$, then we conclude that $G \cong H$. If there is no parse tree of $H$ which is structurally isomorphic to
$T_G$ then $G$ and $H$ can not be isomorphic. 
\par Computing a parse tree $T_G$ of $G$ takes $O(n^2m)$ time. As there are $O(n)$ many labeled graphs in $LabH$,
computing all possible parse trees for labeled graphs in $LabH$ takes 
$O(n^3m)$ time. Testing structural isomorphism between two parse trees need $O(n^3)$ time. Therefore the running time to check 
isomorphism between two prime graphs $G$ and $H$ of clique-width at most three is $O(n^3m)$. \qed
 
\par The correctness of the algorithm follows from Lemma~\ref{st} and Theorem~\ref{ct}.
Lemma~\ref{st} shows that if $G \cong H$ then we can always find two structurally isomorphic parse trees $T_G$ and $T_H$ using the modified 
CHLRR algorithm.
\begin{lemma} \label{st}
 Let $G$ and $H$ be prime graphs with clique-width at most three. If $G \cong_f H$ then 
for every $T_G$ in $parseG$ there is a $T_H$ in $parseH$ such that $T_G$ is {structurally isomorphic} to $T_H$ where 
  $parseG$ and $parseH$ are the set of parse trees generated by Algorithm~\ref{algoparse} on input $LabG$ and $LabH$ 
  respectively.
\end{lemma}
\begin{proof}
If $G \cong_f H$ then from Lemma~\ref{th12} we have $LabG\equiv LabH$ i.e., for every $A$ in $LabG$ there is a $B=g(A)$ in $LabH$ such that $A \cong^{\pi}_f B$ for some $f$ and $\pi$. On input such $A$ and $B$ to 
Lemma~\ref{de} we get two parse trees $T_A$ and $T_B$ which are structurally isomorphic.
\qed
\end{proof}

\begin{theorem} \label{ct}
 Let $G$ and $H$ be graphs with clique-width at most three. Then there exists a polynomial time algorithm to check whether $G \cong H$.
\end{theorem}
\begin{proof}
The proof follows from the prime graph isomorphism of graphs with clique-width at most three described in Lemma~\ref{st} and
Theorem~\ref{primec}. \qed
\end{proof}

\bibliographystyle{splncs03}

\bibliography{myrefs.bib}



\section*{\appendixname}
\section{Graph Isomorphism Completeness for Prime Graphs }
For each vertex $v \in V(G)$, the polynomial-time many-one reduction adds a  
new vertex $v'$ and adds an edge between $v$ and $v'$ to get a new graph $\widehat{G}$. After the addition of vertices and edges to the graph 
it is easy to see that each old vertex in the graph is 
adjacent to exactly one vertex of degree one. It is not hard to see that if $M$ is a non-trivial module in a graph then no vertex in $M$ 
is adjacent to a vertex of degree one. Thus, we can conclude that $\widehat{G}$ is prime graph.
\begin{lemma}\label{PrimeGI}
Given two connected graphs $G_1$ and $G_2$, $G_1 \cong G_2$ iff $\widehat{G}_1 \cong \widehat{G}_2$.
\end{lemma}
\begin{proof}
Let $\widehat{G}_1$ and $\widehat{G}_2$ are graphs obtained after adding new vertices to $G_1$ and $G_2$ respectively. 
If $G_1 \cong_f G_2$ then  we can find an isomorphism between $\widehat{G}_1$ and $\widehat{G}_2$ by extending $f$ to 
newly added vertices such that for every new vertex $y \in  \widehat{G}_1$ having neighbor $x$, $f$ maps  
$y$ to $z$, where $z$ is the newly added neighbor of $f(x)$ in $\widehat{G}_2$.
For the other direction when $\widehat{G}_1\cong_f\widehat{G}_2$, as there are no old vertices of degree one in $\widehat{G}_1$ 
and $\widehat{G}_2$ any isomorphism $f$ from $\widehat{G}_1$ to $\widehat{G}_2$ 
must map the old vertices of $\widehat{G}_1$ to the old vertices of $\widehat{G}_2$. The restriction of $f$ to the old vertices of 
$\widehat{G}_1$ and $\widehat{G}_2$ is an isomorphism from $G_1$ to $G_2$. \qed
\end{proof}

\begin{lemma}
$LabG\equiv LabH$ iff $G \cong_f H$.
\end{lemma}
\begin{proof}
It is easy to see that if $LabG \equiv LabH$ then $G \cong_f H$ from the definition. 
For the other direction, 
given two graphs $G$ and $H$ isomorphic via $f$, 
we need to prove that 
there is a bijection $g: LabG \rightarrow LabH$ such that for all $A \in LabG$, there is  
an isomorphism $f : V(A) \rightarrow V(g(A))$ and a bijection $\pi : lab(A) \rightarrow lab(g(A))$ 
such that $A \cong^{\pi}_f g(A)$.

The proof is divided into five cases based on how 
bilabelings and trilabelings  are generated by CHLRR algorithm described in Section~\ref{sec5}.

\begin{enumerate}
 \item Let $A \in LabG$ be generated at $B_1$ in CHLRR algorithm. Therefore, $A$ has bilabeling $\{v\}$.
  Since $G \cong_f H$, there is a graph $B \in LabH$ which has bilabeling $\{f(v)\}$. 
  Define $g(A)=B$ and a bijection $\pi : lab(A) \rightarrow lab(B)$ such that if $lab(v)=i$ then $\pi(i) = lab(f(v))$ so that $A \cong_f^{\pi} B$.
 \item Let $A \in LabG$ be generated at $B_2$. Thus, $A$ has bilabeling $P= \{x \in N(v) | N[x] \subseteq N[v] \}$.
 As $G \cong_f H$, there is a graph $B \in LabH$ with bilabeling  $f(P)= \{f(x) \in N(f(v)) | N[f(x)] \subseteq N[f(v)] \}$.
 Define $g(A)=B$ and a bijection $\pi : lab(A) \rightarrow lab(B)$ such that if $lab(P)=i$ then $\pi(i) = lab(f(P))$
 so that $A \cong_f^{\pi} B$.
 \item  Let $A \in LabG$ be generated at $T_1$ for a special edge $s$ in the skeleton of $G$ with trilabeling $\tilde{X}$, 
 $\tilde{Y}$, $V(G) \setminus (\tilde{X} \cup \tilde{Y})$.
 As $G \cong_f H$ and the skeleton of graph is unique (from Theorem~\ref{th1}), we can find a $B \in LabH$ which is generated for the 
 special edge $f(s)$ in skeleton of $H$ which corresponds to trilabeling $f(\tilde{X})$, 
 $f(\tilde{Y})$, $V(H) \setminus (f(\tilde{X} \cup \tilde{Y}))$.
 Define $g(A)=B$ and a bijection $\pi : lab(A) \rightarrow lab(B)$ such that if $lab(\tilde{X})=i_1$ then $\pi(i_1) = lab(f(\tilde{X}) )$,
 if $lab(\tilde{Y})=i_2$ then $\pi(i_2) = lab(f(\tilde{Y}) )$, and  if $lab(V(G) \setminus (\tilde{X} \cup \tilde{Y}))=i_3$ 
 then $\pi(i_3) = lab(V(H) \setminus (f(\tilde{X} \cup \tilde{Y})))$
so that $A \cong_f^{\pi} B$.
 \item Let $A \in LabG$ be generated at $B_3$ for a clique component $C$ with bilabeling $C$. 
 As $G \cong_f H$, there is a $B \in LabH$ which is  
 generated for a clique component $f(C)$ with bilabeling $f(C)$. 
 Define $g(A)=B$ and a bijection $\pi : lab(A) \rightarrow lab(B)$ such that if $lab(C)=i$ then 
 $\pi(i) = lab(f(C))$ so that $A \cong_f^{\pi} B$.
 \item Let $A \in LabG$ be generated at $B_4$ for a star component $S$ with bilabeling $\{c\}$, where 
 $c$ is a special center of $S$. 
 As $G \cong_f H$, there is a graph $B \in LabH$ which is  
 generated for a star component $f(S)$ with bilabeling $f(c)$, where 
 $f(c)$ is a special center of $f(S)$. 
 Define $g(A)=B$ and a bijection $\pi : lab(A) \rightarrow lab(B)$ such that if $lab(c)=i$ then 
 $\pi(i) = lab(f(c))$ so that $A \cong_f^{\pi} B$. \qed
 \end{enumerate}

\end{proof}

\section{Generating Structurally Isomorphic Parse Trees} \label{sec7}
In this section we prove that the modified CHLRR algorithm generates structurally isomorphic parse trees on two isomorphic input graphs. 
To prove that we also show that the supporting subroutines do the same.

\begin{algorithm}[H]
{\scriptsize 
\DontPrintSemicolon
\caption{Finding parse trees for labeled graphs of clique-width at most three in $LabG$}
\label{algoparse}
\KwIn{$LabG$ a set of bilabelings and trilabelings of $G$}
\KwOut{ $parseG=\{T_A \hspace{0.1cm}|\hspace{0.1cm} T_A$ is a parse tree of graph $A$ in LabG  of clique-width 
at most three $\}$ }
\Begin{
$parseG:=\emptyset$\;
\For {{\bf all} $A \in LabG$}{
$A.parse$-$tree:=null$\;
$A.parse$-$tree = Decompose(A)$\;
\If{$A.parse$-$tree \neq null$}
{
Add $A.parse$-$tree$ to $parseG$\;
}

}
return ($parseG$)
}
}
\end{algorithm}
The function $Decompose(P)$ in Algorithm~\ref{algoparse} finds parse tree of $P$ if $cwd(P) \leq 3$ and it is described in following Section and Appendix. \\


\begin{algorithm}
{\scriptsize 
\DontPrintSemicolon
\caption{Function Decompose \cite{corneil2012polynomial}}
\label{algodecompose}
\KwIn{A bi or trilabled $l$-prime connected graph $P$}
\KwOut{A parse tree of $P$ or null parse tree if $cwd(P)$ $>$ 3}
\Begin{
$parse$-$tree:=$ a trivial parse tree with $P$ as the unique leaf\;
\tcc{$parse$-$tree$ may contain connected components as leafs but as the algorithm proceeds this components will be decomposed to 
finally obtain the parse tree}
$Leaves := \{P\} $ \tcc*{$Leaves$ contains pointer to $P$} 
\While{$Leaves \neq \emptyset$}{
flag := true,  tree := null\;
Extract $\Gamma$ from $Leaves$\;
\If{ $\Gamma$ has no more than three vertices}{
Find a canonical parse tree, $tree$ \;
Replace $\Gamma$ by $tree$ in $parse$-$tree$\;
}
\eIf{$\Gamma$ is trilabled}{
[flag, tree] = Decompose-leaf-TI($\Gamma$)\;
Add the leafs of $tree$ to $Leaves$\;
Replace $\Gamma$ by $tree$ in $parse$-$tree$\;
}
{
[flag, tree] = Decompose-leaf-BI($\Gamma$)\;
Add the leafs of $tree$ to $Leaves$\;
Replace $\Gamma$ by $tree$ in $parse$-$tree$\;
}
\If{flag is false}{
$parse$-$tree$ := null\;
return ($parse$-$tree$)\;
}
}
return ($parse$-$tree$)\;
}
}
\end{algorithm}
\subsection{Decomposing Trilabeled Graphs}
The function $Decompose$-$leaf$-$TI$ (Algorithm~\ref{trilabel}) decomposes trilabeled graph from $LabG$. It can be check that this 
function is always called with inputs coming from $LabG$. In  other words it is only called in the first level of the recursion.

\begin{lemma} \label{TI}
Let $A$ in $LabG$ and $B$ in $LabH$ be trilabeled and $l$-prime connected graphs. 
If $A \cong _f^{\pi} B$ for some $f$ and $\pi$ then Algorithm~\ref{trilabel} generates 
 top operations of parse trees for $A$ and $B$ such that 
 $\pi$ $\in$ $\ISO(Q_a,Q_b)$ with 
$A_a \cong^{\pi}_f B_b$, 
where $A_a$ and $B_b$ are the graphs described in Algorithm~\ref{trilabel}.
\end{lemma}
\begin{proof}
Let $A$ and $B$ are trilabeled with $l_1$, $l_2$, $l_3$ and $l'_1$, $l'_2$, $l'_3$ respectively.  
If $A$ has a trivial split (see Figure~\ref{TS}) then it has a universal vertex $x$ of some label $l_1$.
Then the algorithm removes the edges $E_{l_1l_2}$, $E_{l_1l_3}$ from $A$ 
 and gives a decomposition with top operations $\eta_{l_1,l_2}$ and $\eta_{l_1,l_3}$ above a $\oplus$ operation whose children are
$x$ and connected components $A_{a_1}, \cdots ,A_{a_k}$ ($A_a = x \oplus A_{a_1}\oplus \cdots \oplus A_{a_k}$). 
If $A \cong_f^{\pi} B$, then there is a universal vertex $y$ in $B$ of label $\l'_1$ such that $f(x)=y$ and $\pi(l_1)=l'_1$.
To decompose $B$, the algorithm removes the edges $E_{l'_1l'_2}$, 
$E_{l'_1l'_3}$ from $B$ to get the decomposition with top operations $\eta_{l'_1l'_2}$ and $\eta_{l'_1l'_3}$ above a 
$\oplus$ operation whose children are $y$ and connected components $B_{b_1}, \cdots ,B_{b_k}$ ($B_b = y \oplus B_{b_1}\oplus \cdots \oplus B_{b_k}$).
In fact $B_{b_1}, \cdots ,B_{b_k}$ are images of $A_{a_1}, \cdots ,A_{a_k}$ under $f$ in some order.
The quotient graphs $Q_a$ and $Q_b$ have three vertices corresponding to top two consecutive $\eta$ operations. 
If $A \cong_{f}^{\pi} B$ the quotient graphs are isomorphic via $\pi$
and $A_a \cong^{\pi}_fB_b$. 
\begin{algorithm}
{\scriptsize 
\DontPrintSemicolon
\caption{Function Decompose-leaf-TI \cite{corneil2012polynomial}}
\label{trilabel}
\KwIn{A trilabeled, $l$-prime and connected graph $G$}
\KwOut{true with top operations of parse tree or false if $cwd(G)$ $>$ 3}
\Begin{
tree := null\; 
\If{$G$ has a universal vertex $x$ of label $l_1$}  
{ 
Let $G_g =x \oplus^k_{i=1}G_{g_i}$, where $x,G_{g_i}$'s are connected components of $G \setminus \{E_{l_1l_2},E_{l_1l_3}\}$ \tcc*{$l_1,l_2,l_3$ are labels of $G$} 
$tree= \eta_{l_1,l_2} \eta_{l_1,l_3}(x \oplus^k_{i=1}G_{g_i})$ \;
return (true, tree)\;
}
\If{$G$ has two labels $l_1,l_2$ such that $E_{l_1l_2}$ is complete}
{Let $G_g =\oplus^k_{i=1}G_{g_i}$, where $G_{g_i}$'s are connected components of $G \setminus \{E_{l_1l_2}\}$\; 
$tree= \eta_{l_1,l_2}(\oplus^k_{i=1}G_{g_i})$ \;
return (true, tree)\;
}
return (false, tree) (i.e., $cwd(G)$ $>$ 3)
}
}
\end{algorithm}

\begin{figure}[!ht]
\centering
\includegraphics[trim=-1cm 21cm 5cm 3cm, clip=true, scale=0.8]{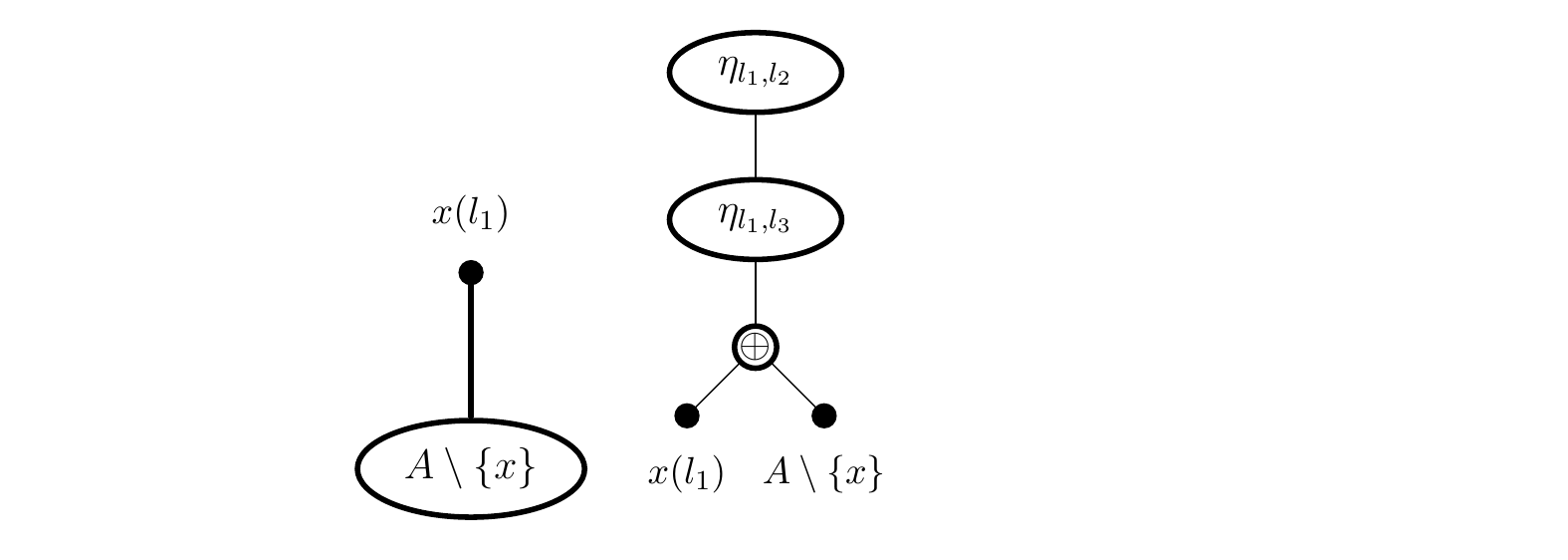}
\caption{ \hspace{0.1cm} Trivial split: $x$ is a universal vertex of label $l_1$ in a trilabeled graph $A$.
We use the bold edge between two sets of vertices to indicate that all
edges are present between two vertex sets.}
\label{TS} 
\end{figure}

\begin{figure}[!ht]
\centering
\includegraphics[trim=5.2cm 22cm 3cm 3cm, clip=true, scale=0.8]{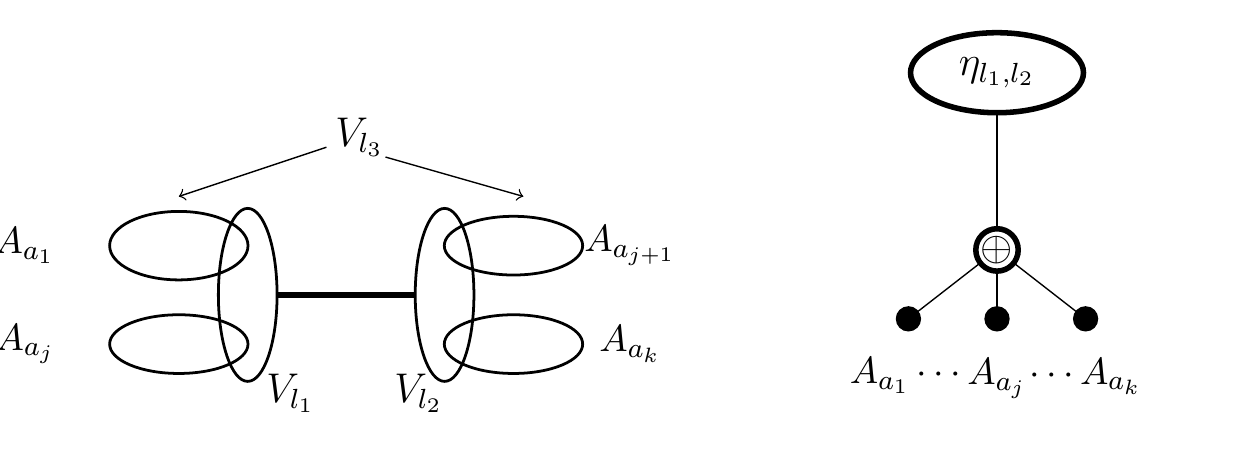}
\caption{ \hspace{0.1cm} Nontrivial split: $V_i$ represents set of vertices in $A$ that have label $i$.  }
\label{NS} 
\end{figure}

\par If $A$  corresponds to a  nontrivial split (see Figure~\ref{NS}) then there are two labels $l_1,$ $l_2$ such that $E_{l_1l_2}$ is 
complete. We get a decomposition with $\eta_{l_1,l_2}$ 
operation above a $\oplus$ operation whose children are
connected components $A_{a_1}, \cdots ,A_{a_k}$ ($A_a =  A_{a_1}\oplus \cdots \oplus A_{a_k}$) of $A$ 
after the  $E_{l_1l_2}$ edges are removed.  
If $A \cong_f^{\pi}B$, then there exists a nontrivial split in $B$ and two labels $l'_1,$ $l'_2$ such that $E_{l'_1l'_2}$ is 
complete and $\pi\{l_1,l_2\}=\{l'_1,l'_2\}$, $\pi(l_3)=l'_3$. To decompose $B$, the algorithm removes the edges $E_{l'_1l'_2}$, 
to get the decomposition with top operations $\eta_{l'_1,l'_2}$  
above a $\oplus$ operation whose children are connected components $B_{b_1}, \cdots ,B_{b_k}$ obtained 
from $B$ after $E_{l'_1l'_2}$ edges are removed.
The quotient graphs $Q_a$ and $Q_b$ build from the top 
operations are isomorphic via $\pi$
and $A_a \cong^{\pi}_fB_b$. \qed
\end{proof}

\begin{algorithm}
{\scriptsize 
\DontPrintSemicolon
\caption{Function $Decompose$-$leaf$-$BI$ (cf., \cite{corneil2012polynomial})}
\label{algodecompose-leaf-BI}
\KwIn{A bilabeled, $l$-prime and connected graph $G$}
\KwOut{true with top operations of parse tree or false if $cwd(G)$ $>$ 3}
\Begin{
tree := null\; 
\If {$G \in$ PC1}{
\If {$G$ has a universal vertex (say $x$)}{
Let $G_g =x \oplus^k_{i=1}G_{g_i}$, where $x, G_{g_i}$'s are connected components of $G \setminus \{E_{l_3l_2},E_{l_3l_1}\}$ \tcc*{$l_1,l_2$ are labels of $G$} 
$tree= \rho_{l_3\rightarrow l}\eta_{l_3,l_2}\eta_{l_3,l_1}(x \oplus^k_{i=1}G_{g_i})$\;
return (true, tree)\;
}
}

\If {$G \in$ PC2}{
Compute a set $S$ number of vertices in $G$ which are universal to one label class but not adjacent to other label 
class\; 
\If {$|S|$ equal to 1(say $x$)}{
Let $G_g =x \oplus^k_{i=1}G_{g_i}$, where $x, G_{g_i}$'s are connected components of $G \setminus \{E_{l_3l'}\}$\; 
$tree= \rho_{l_3\rightarrow l}\eta_{l_3,l'}(x \oplus^k_{i=1}G_{g_i})$\;
return (true, tree)\;
}
\If {$|S|$ equal to 2 (say $x_1$ and $x_2$)}{
Let $G_g =x_1 \oplus x_2 \oplus^k_{i=1}G_{g_i}$, where $x_1, x_2, G_{g_i}$'s are connected components of $G \setminus \{E_{l_3l'},E_{l_4,\overline{l'}}\}$\; 
$tree= \rho_{l_4\rightarrow \overline{l'}}\eta_{l_4,\overline{l'}}\rho_{l_3 \rightarrow l'}\eta_{l_3,l'}(x_1\oplus x_2 \oplus^k_{i=1}G_{g_i})$\;
return (true, tree)\;
}
}
\If {$G \in$ PC3}{
Let $G_g =x\oplus y \oplus^k_{i=1}G_{g_i}$, where $x,y,G_{g_i}$'s are connected components of $G \setminus \{E_{l_3l},E_{l_4,\bar{l}}\}$\; 
$tree= \rho_{l_3\rightarrow l}\eta_{l_3,l}(x \oplus \eta_{l_3,\bar{l}}(y \oplus^k_{i=1}G_{g_i})) $\;
return (true, tree)\;
}

Compute the coconnected components of $V_{l_1}$ and $V_{l_2}$ and test membership of $G$ in 
$\mathcal{U}_{l_1},\mathcal{U}_{l_2}, \mathcal{\overline{D}}_{l_1}$ and $\mathcal{\overline{D}}_{l_2}$\;
{\bf if} {$G \in \mathcal{U}_{l_1}$ } {\bf then} {return(Decompose-leaf-$\mathcal{U}_{l_1}(G)$)\;}
{\bf if} {$G \in \mathcal{U}_{l_2}$ } {\bf then} {return(Decompose-leaf-$\mathcal{U}_{l_2}(G)$)\;}
{\bf if} {$G \in \mathcal{\overline{D}}_{l_1}$ } {\bf then} {return(Decompose-leaf-$\mathcal{\overline{D}}_{l_1}(G)$)\;}
{\bf if} {$G \in \mathcal{\overline{D}}_{l_2}$ } {\bf then} {return(Decompose-leaf-$\mathcal{\overline{D}}_{l_2}(G)$)\;}
{\bf if} {$G \notin \mathcal{U}_{l_1},\mathcal{U}_{l_2}, \mathcal{\overline{D}}_{l_1}$ and $\mathcal{\overline{D}}_{l_2}$ } {\bf then}
{return(false,tree) (i.e., $cwd(G)$ $>$ 3)\;}
}
}
\end{algorithm}

\subsection{Decomposing Bilabeled Graphs}
Our modification to the CHLRR algorithm is in $Decompose$-$leaf$-$BI$ (Algorithm~\ref{algodecompose-leaf-BI}), where we use 
four labels instead of three to find structural isomorphic parse trees.
If $G$ is a bilabeled, $l$-prime and connected graph of clique-width at most three, then either $ G \in PCi$ where $i\in \{1,2,3\}$ or 
$G \in \mathcal{U}_i$, $\mathcal{\overline{D}}_i$ where $i \in \{1,2\}$ (See Proposition 29 in \cite{corneil2012polynomial}). From here on wards we assume that $G$ and $H$ are bilabeled with $l_1$, $l_2$ and
$l'_1$, $l'_2$ respectively.  
\begin{lemma} \label{BI}
Let $G$ and $H$ be bilabeled, $l$-prime and connected graphs. If $G \cong _f^{\pi} H$ for some $f$ and $\pi$ 
then Algorithm~\ref{algodecompose-leaf-BI} 
 generates top operations of parse trees $G$ and $H$
such that there is a $\pi_i$ $\in$ $\ISO(Q_g,Q_h)$ with
$G_g \cong^{\pi_i}_f H_h$ and $\pi_i/color = \pi|_{color(Q_g)}$,  
where $G_g$ and $H_h$ are the graphs described in  Algorithm~\ref{algodecompose-leaf-BI}.
\end{lemma}
\begin{proof}
 There are three simple cases that can be handled easily. These simple cases denoted as PC1, PC2 and PC3.
 The other cases $\mathcal{U}_{l_1}$ $(\mathcal{U}_{l_2})$ and $\mathcal{\overline{D}}_{l_1}$ $(\mathcal{\overline{D}}_{l_2})$ are 
 described in Algorithms~\ref{u} and \ref{D} in Appendix.
\par {\bf PC1:}
If $G \in PC1$ then $G$ has a universal vertex of label $l \in \{l_1,l_2\}$ (see Figure~\ref{PC1}).
Note that in this case $G$ can not  have more than two universal vertices of same label, otherwise those universal vertices form an $l$-module.
\par 
To decompose $G$ the algorithm relabels vertex $x$ with $l_3$
and removes the edges $E_{l_3l_2}$ and $E_{l_3l_1}$. Then we get the decomposition with $\rho_{l_3\rightarrow l}$, $\eta_{l_3,l_2}$, $\eta_{l_3,l_1}$
above a $\oplus$ operation with children
$x$ and connected components $G_{g_1}, \cdots ,G_{g_k}$ ($G_g = x \oplus G_{g_1}\oplus \cdots \oplus G_{g_k}$). 
If $G \cong_f^{\pi} H $ then algorithm finds the unique 
universal vertex $y$ in $H$ of label $l' \in \{l'_1, l'_2\}$ such that $f(x)=y$ and $\pi(l) = l'$. To decompose $H$ the algorithm  relabels the vertex $y$ with $l'_3$ and removes the edges 
$E_{l'_3 l'_2}$ and $E_{l'_3 l'_1}$ to get the decomposition with $\rho_{l'_3\rightarrow l'}$, $\eta_{l'_3,l'_2}$, $\eta_{l'_3,l'_1}$ 
above a $\oplus$ operation with children $y$ and connected components $H_{h_1}, \cdots ,H_{h_k}$ 
(these are images of $G_{g_1}, \cdots ,G_{g_k}$ under $f$ in some order). 
The quotient graphs $Q_g$ and $Q_h$ build from the top 
operations are isomorphic via $\pi_i$, where $\pi_i (l_3) = l_3'$ and $\pi_i(l)= \pi(l)$ if $l \in \{l_1,l_2\}$. It is clear that,
$G_g \cong^{\pi_i}_f H_h$ and $\pi_i/color = \pi|_{color(Q_g)}$.
\begin{figure}[!ht]
\centering
\includegraphics[trim=3cm 19.8cm 12cm 3cm, clip=true, scale=0.8]{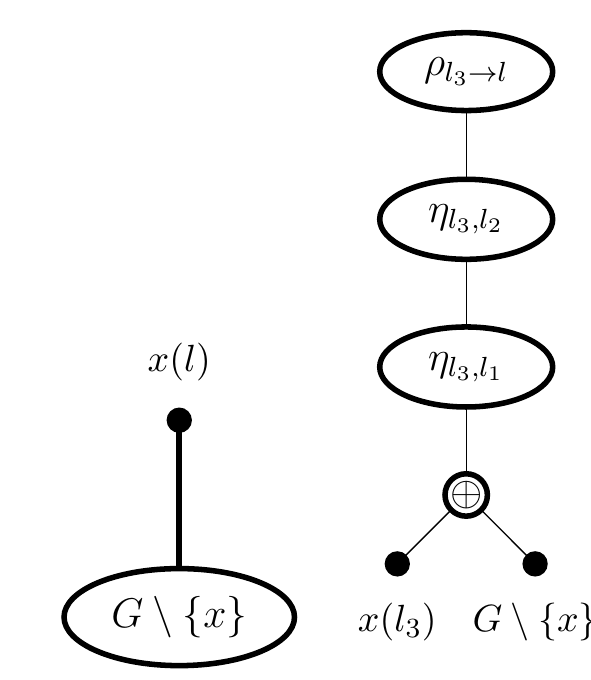}
\caption{\hspace{0.1cm} PC1: Decomposition of a bilabeled graph $G$ with a universal vertex $x$.  }
\label{PC1} 
\end{figure}

\par Suppose $G$ has two universal vertices $x_1$ and $x_2$ of label $l_1$ and $l_2$ respectively. 
In this case we apply above procedure consecutively two times first taking $x_1$ as a universal vertex in graph $G$, second taking $x_2$ as a 
universal vertex in graph $G \setminus \{x_1\}$. Note that the order in which we consider $x_1$ and $x_2$ does not effect the structure of the parse 
tree.

\par {\bf PC2:}
If $G \in PC2$ then $G$ can have one or two vertices of different labels which are universal to vertices of one 
label class but not to other label class.
Let $l_1$ and $l_2$ be the labels of $G$.
In this case the algorithm finds the decomposition of $G$ described as follows:
\par {\it Case-1:} Suppose $G$ has a single vertex $x$ of label $l$ (see Figure~\ref{PC21}) that is universal to all vertices of label $l' \in \{l_1,l_2\}$, but not 
adjacent to all vertices of label $\bar{l'} \in \{l_1,l_2\} \setminus l'$. 
To decompose $G$, the algorithm relabels $x$ with a label $l_3 \notin \{l_1,l_2\}$ and removes the edges $E_{l_3l'}$, which gives the decomposition with 
top operations $\rho_{l_3\rightarrow l}$, $\eta_{l_3,l'}$ above a $\oplus$ operation with children
$x$ and connected components $G_{g_1}, \cdots ,G_{g_k}$ ($G_g = x \oplus G_{g_1}\oplus \cdots \oplus G_{g_k}$).
If $G \cong_f^{\pi} H $, the algorithm finds a vertex $y$ in $H$ of label $m$ which is universal to all vertices of label $m'$ but not adjacent to 
all vertices of label $\bar{m}'$ such that $f(x)=y$ and $\pi(l)=m$.
To decompose $H$ the algorithm relabels $y$ with a label $l'_3 \notin \{l'_1, l'_2 \}$ and 
removes the edges $E_{l'_3 m}$, which gives the decomposition with top operations
$\rho_{l'_3 m}$, $\eta_{l'_3,m'}$ 
above a $\oplus$ operation whose children are $y$ and the connected components $H_{h_1}, \cdots ,H_{h_k}$ (these are images of 
$G_{g_1}, \cdots ,G_{g_k}$ under $f$ in some order). 
The quotient graphs $Q_g$ and $Q_h$ build from top 
operations are isomorphic via $\pi_i$, where $\pi_i (l) = m $, $\pi(\bar l)= \bar m$  and $\pi_i (l_3)=l'_3$. 
Moreover, $G_g \cong^{\pi_i}_f H_h$ and $\pi_i/color = \pi|_{color(Q_g)}$.
\begin{figure}[!ht]
\begin{subfigure}[t]{0.5\textwidth}
\centering
  \includegraphics[trim=3cm 20cm 11cm 1cm, clip=true, scale=0.8]{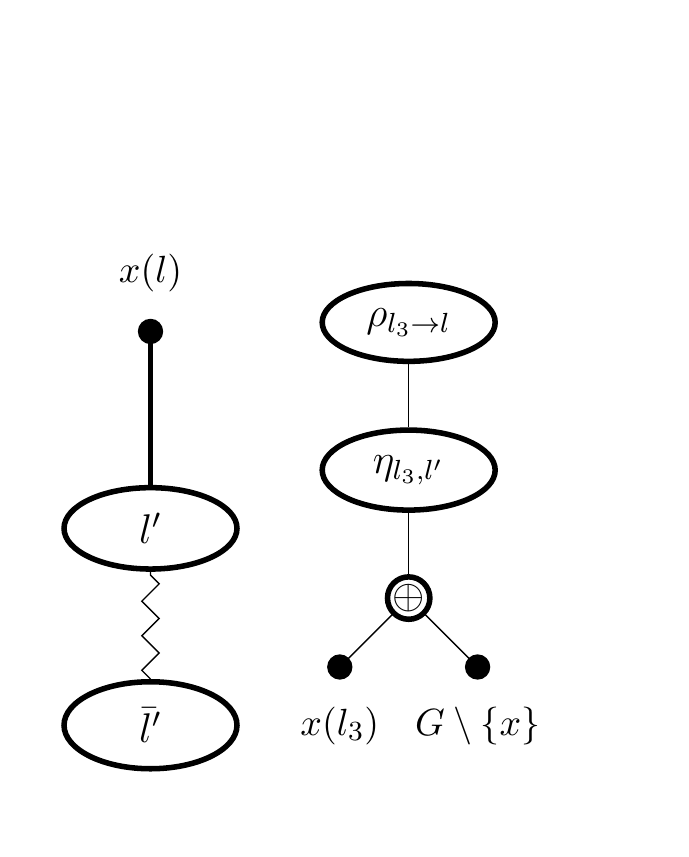}
\caption{PC2: Case 1 }
\label{PC21} 
\end{subfigure}
\begin{subfigure}[t]{0.5\textwidth}
 \centering
  \includegraphics[trim=3cm 18cm 9cm 3cm, clip=true, scale=0.8]{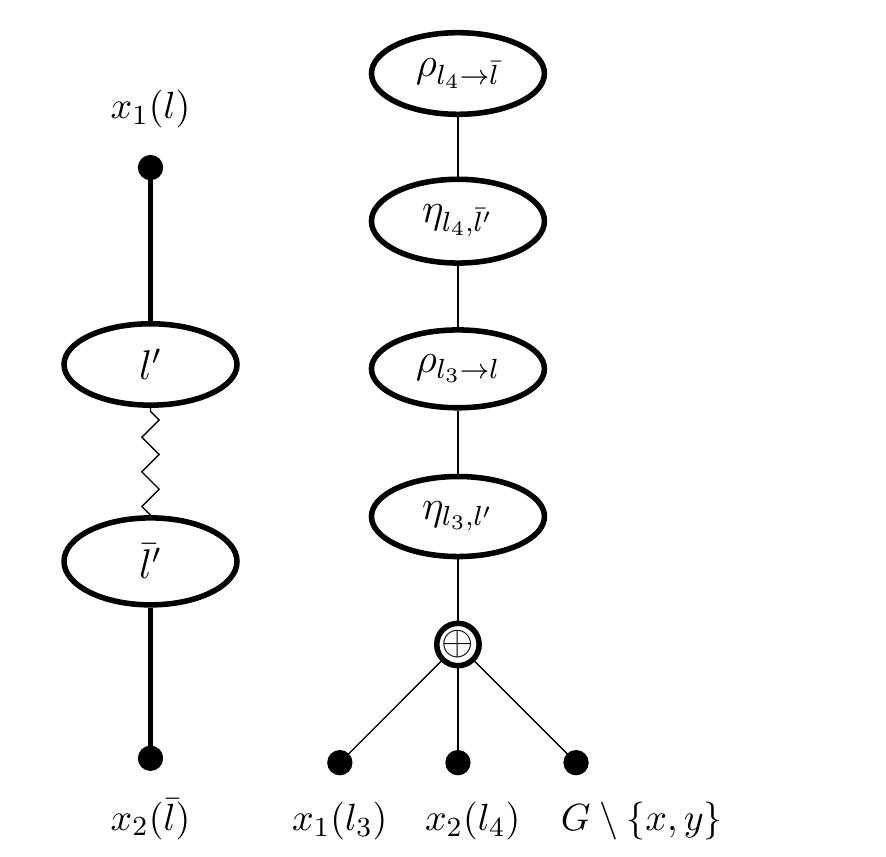}
\caption{PC2: Case 2}
\label{PC22} 
\end{subfigure}
\caption { \hspace{0.1cm} 
Decomposing a bilabeled graph $G$, having one or two vertices of different labels
which are universal to vertices of one label class but not to other label class.
We use the zigzag edge to indicate the presence of some edges
between the two sets of vertices}
\end{figure}

{\it Case-2:} Suppose $G$ has two vertices $x_1$ and $x_2$ of label $l \in \{l_1,l_2\}$ (see Figure~\ref{PC22}) and $\bar l \in \{l_1,l_2\} \setminus l$ such that $x_1$ ($x_2$) is universal to all vertices of label $l'\in \{l_1,l_2\}$ $(\bar l' )$, but not adjacent to all vertices of label ${\bar l'}$ $(l')$.
Then the algorithm relabels vertices $x_1$ and $x_2$ with $l_3$ and $l_4$ respectively and removes edges $E_{l_4,{\bar l'}}$, $E_{l_3,l'}$
to get the decomposition of $G$ with $\rho_{l_4\rightarrow {\bar l}}$, $\eta_{l_4,{\bar l'}},\rho_{l_3 \rightarrow l}$, $\eta_{l_3,l'}$
above a $\oplus$ operation with children
$x_1$, $x_2$  and connected components $G_{g_1}, \cdots ,G_{g_k}$ ($G_g = x_1 \oplus x_2 \oplus  G_{g_1}\oplus \cdots \oplus G_{g_k}$).
If $G \cong_f^{\pi} H $, the algorithm finds vertices $y_1$ and $y_2$ in $H$ of label $m \in \{l'_1,l'_2\}$ and 
$\bar m \in \{l'_1,l'_2\} \setminus m$ such that 
$y_1 (y_2)$ is universal to all vertices
of label $m'$ ($\bar m'$), but not adjacent to all the vertices of label $\bar m'$ $(m')$ and $f(x_1)=y_1$, $f(x_2)=y_2$.
Then algorithm relabels vertices $y_1$ and $y_2$ with $l'_3$ and $l'_4$ respectively and  removes edges $E_{l'_4,{\bar m'}}$,$E_{l'_3,m'}$
to get the decomposition of $H$ with top operations
$\rho_{l'_4\rightarrow {\bar m}}$, $\eta_{l'_4,{\bar m'}},\rho_{l'_3 \rightarrow m}$, $\eta_{l'_3, m'}$ 
above a $\oplus$ 
operation whose children are $y_1$, $y_2$ and connected components $H_{h_1}, \cdots ,H_{h_k}$ (these are images of 
$G_{g_1}, \cdots ,G_{g_k}$ under $f$ in some order). 
The quotient graphs $Q_g$ and $Q_h$ build from the top 
operations are isomorphic via $\pi_i$, where
$\pi_i (l') = \pi(l')= m'$, $\pi_i (\bar l') = \pi(\bar l')= \bar m'$, $\pi_i(l_3)=l'_3$ and $\pi_i(l_4)=l'_4$.
It is clear that, $G_g \cong^{\pi_i}_f H_h$ and $\pi_i/color = \pi|_{color(Q_g)}$.

\par {\bf PC3:} If $G \in PC3$ then $G$ has two vertices $x$ and $y$ of label $l$, where $y$ is universal to everything other than $x$, 
and $x$ is universal to all vertices of label $l$ other than $y$, and non-adjacent of all vertices of the other label $\overline{l}$ as 
shown in Figure~\ref{PC3}.
To decompose $G$ the algorithm relabels the vertices $x$ and $y$ with $l_3$ and removes the edges $E_{l_3l}$ to get the decomposition of $G$ with 
top operations $\rho_{l_3\rightarrow l}$, $\eta_{l_3,l}$ and a $\oplus$
with the connected components of $G_{g}=x \oplus G \setminus \{x\}$ as children.
Again the algorithm removes the edges $E_{l_3\overline{l}}$ from $G \setminus \{x\}$ to get the decomposition 
with top operations $\eta_{l_3,\bar{l}}$ and a $\oplus$ with the connected components of $G_{g_1}=y\oplus G \setminus \{x, y\}$ as children.
If $G \cong_f^{\pi} H$, the algorithm finds the vertices $x', y' \in H$ of label $l'$ such that $f(x)=x'$, $f(y)=y'$ and
$\pi(l)=l'$.
where $y'$ is universal to everything other than $x'$, 
and $x'$ is universal to all vertices of label $l'$ other than $y'$, and non-adjacent to all vertices of label $\overline{l}'$.
Then it relabels vertices $x'$ and $y'$ with $l'_3$ and removes the edges $E_{l'_3l'}$ to get the decomposition of $H$ with top operations 
$\rho_{l'_3\rightarrow l'}$, $\eta_{l'_3,l'}$ and a $\oplus$ 
with the connected components of $H_h=x'\oplus H \setminus \{x'\}$ as children. Again the algorithm removes the edges 
$E_{l'_3, \bar{l}'} $ from $H \setminus \{x'\}$ 
to get the decomposition with top operations
$\eta_{l'_3, \bar{l}'}$ and a $\oplus$ with the connected components of $H_{h_1}=y' \oplus H \setminus \{x', y'\}$ as children.
In this case the generated parse tree has two levels.

\begin{figure}[!ht]
\centering
\includegraphics[trim=3cm 19cm 6cm 3.2cm, clip=true, scale=0.8]{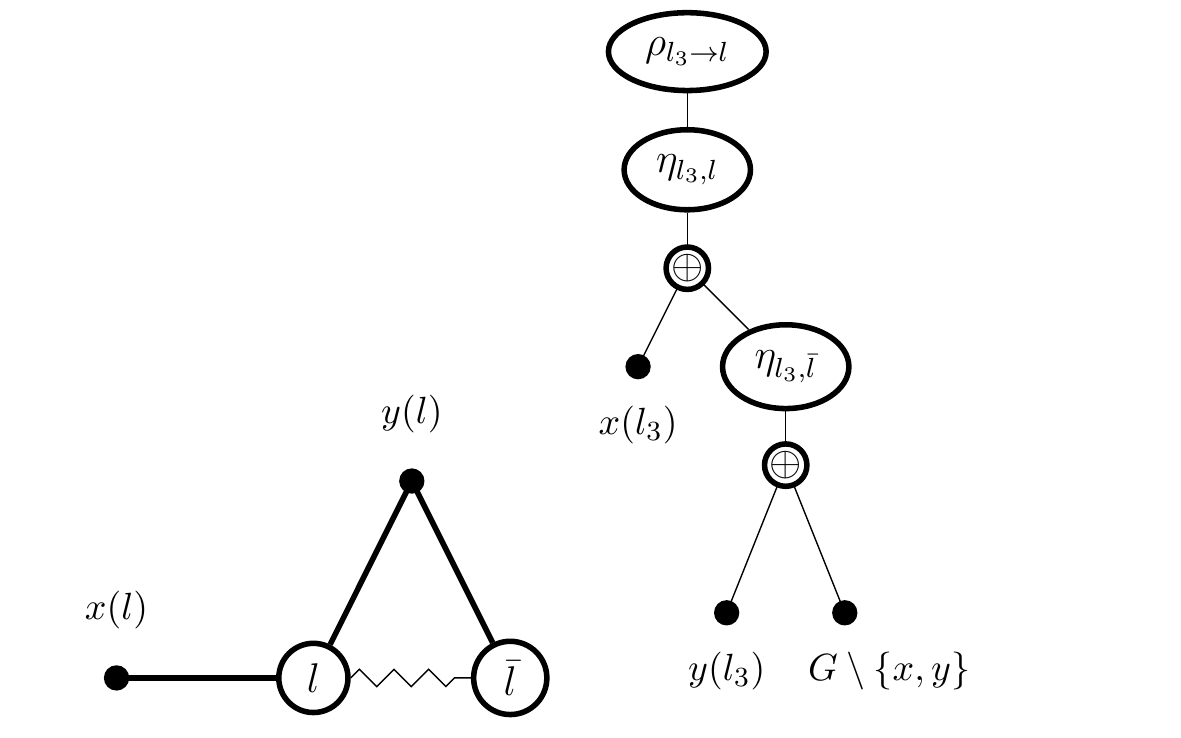}
\caption{ \hspace{0.1cm} Decomposing a bilabeled graph $G$ which has two vertices $x$ and $y$ of label $l$, where
$y$ is universal to everything other than $x$, and $x$ is universal to all vertices of label $l$ 
other than $y$, and non-adjacent of all vertices of the other label $\bar l$}
\label{PC3} 
\end{figure}
In first level the quotient graphs $Q_g$ and $Q_h$ build from top 
operations are isomorphic via $\pi_{1i}$, where $\pi_{1i} (l) = \pi(l)$ if $l \in \{l_1,l_2\}$, $\pi_{1i}(l_3)=l'_3$.
It is clear that, $G_g \cong^{\pi_{1i}}_f H_h$ and $\pi_{1i}/color = \pi|_{color(Q_g)}$.
In second level the quotient graphs $Q_{g_1}$ and $Q_{h_1}$ build from top
operations are isomorphic via $\pi_{2i}$, where  $\pi_{2i}(l) = \pi_{1i}(l)$ if $l \in \{l_1,l_2,l_3\}$,  
and $G_{g_1} \cong^{\pi_{2i}}_f H_{h_1}$, $\pi_{2i}/color = \pi_{1i}|_{color(Q_{g_1})}$.
The remaining part of proof follows from Lemma \ref{lu1} and \ref{lD1}. \qed
\end{proof}

\subsection{Function Decompose $\mathcal{U}_{l_1}$:}
We next describe the case $\mathcal{U}_{l_1}$. The case $\mathcal{U}_{l_2}$ is omitted from here because it is similar to $\mathcal{U}_{l_1}$.
Let $l_1$ and $l_2$ be the vertex labels. The vertex set $V_{l_1}$ consisting of vertices with label $l_1$ can be partitioned as follows: 
The set of vertices adjacent to all vertices of $V_{l_2}$ is denoted $V_{l_1}^a$.
The set of vertices adjacent to some of vertices of $V_{l_2}$ is denoted $V_{l_1}^s$.
The set of vertices adjacent to none of vertices of $V_{l_2}$ is denoted $V_{l_1}^n$.
Similarly we can define the sets $V_{l_2}^a$, $V_{l_2}^s$ and $V_{l_2}^n$.
For $l \in \{l_1,l_2\}$ we say, $G \in \mathcal{U}_l$ if $V_{l}^a \neq \emptyset$ and removing the edges between $V_{l}^a$ and $V_{l}$ disconnects $G$.
\begin{algorithm}
{\scriptsize 
\DontPrintSemicolon
\caption{Function $Decompose$-$leaf$-$\mathcal{U}_{l_1}$ \cite{corneil2012polynomial}}
\label{u}
\KwIn{A bilabeled, $l$-prime and connected graph $G$}
\KwOut{true with top operations of parse tree or false if $cwd(G)$ $>$ 3}
\Begin{
tree := null\;
\If{$G$ has good non partial connected component $C$}
{
Let $G_g =\oplus^k_{i=1}G_{g_i}$, where $G_{g_i}$'s are connected components of $G \setminus \{E_{l_3,l_2}\}$\; 
$tree= \rho_{l_3 \rightarrow l_1}\eta_{l_3,l_2}(\oplus^k_{i=1}G_{g_i})$\;
return (true, tree)\;
}
\If { $G$ has only partial connected components }
{
Let $G_g =\oplus^k_{i=1}G_{g_i}$, where $G_{g_i}$'s are connected components of $G \setminus \{E_{l_3,l_2}\}$\; 
$tree= \rho_{l_3 \rightarrow l_1}\eta_{l_3,l_2}(\oplus^k_{i=1}G_{g_i})$\;
return (true, tree)\;
}
return (false, tree) (i.e., $cwd(G)$ $>$ 3)
}
}
\end{algorithm}

\begin{lemma} \label{lu1}
Let $G$ and $H$ be bilabeled, $l$-prime and connected graphs. If $G \cong _f^{\pi} H$ for some $f$ and $\pi$ then Algorithm~\ref{u} 
 generates top operations of parse trees $G$ and $H$ such that there is a $\pi_i$ $\in$ $\ISO(Q_g,Q_h)$ with
$G_g \cong^{\pi_i}_f H_h$ and $\pi_i/color = \pi|_{color(Q_g)}$,  
where $G_g$ and $H_h$ are the graphs described in Algorithm~\ref{u}.
\end{lemma}
\begin{proof}
$G \in \mathcal{U}_{l_1}$ if $V_{l_1}^a \neq \emptyset$ and removing the edges between $V_{l_1}^a$ and $V_{l_2}$ disconnects $G$.
 The proof is divided into two cases based on connected components (partial\footnote {A connected component of $V_1$ that contains at least
one vertex of $V^s_1$ is called partial. \label{note2}} and non partial) of $V_1$. 
 \par If there is at least one \emph{good connected component}\footnote{A non-partial 
 connected component $C$ of $V_1$ is good (respectively, bad), if $G$
is of clique-width at most three implies that the bilabeled graph obtained from $C$ by relabeling all the vertices of $V^a_1 \cup C$ with three is
of clique-width at most three (respectively of clique-width more than three).} $C$
(see Section 5.2.1 in \cite{corneil2012polynomial}) in $G$ then 
 the algorithm relabels all vertices of $V_{l_1}^a$ in good connected components with $l_3$ and removes the edges $E_{l_3l_2}$ 
 from $G$ to get the decomposition with top operations $\rho_{l_3 \rightarrow l_1}$ and $\eta_{l_3,l_2}$  above a $\oplus$ operation with 
 the connected components $G_{g_1}, \cdots ,G_{g_k}$ as children ($G_g = G_{g_1}\oplus \cdots \oplus G_{g_k}$). 
 If $G \cong_f^{\pi} H$, up to a permutation of 
 labels $H$ may be in $ \mathcal{U}_{l'_1}$ or $ \mathcal{U}_{l'_2}$, but this does not effect the structure of the decomposition as in both the cases 
 the set of edges deleted are same.
 The algorithm finds at least one good connected component $C'$ in $H$ and relabels all vertices of $V_{l'_1}^a$ in good connected components 
 with $l'_3$ and removes the edges $E_{l'_3l'_2}$ from $H$ to get the decomposition 
 with top operations $\rho_{l'_3 \rightarrow l'_1}$ and $\eta_{l'_3,l'_2}$ 
 above a $\oplus$ operation with connected components $H_{h_1}, \cdots ,H_{h_k}$ as children (these are images of 
$G_{g_1}, \cdots ,G_{g_k}$ under $f$ in some order). 
The quotient graphs $Q_g$ and $Q_h$ build from top 
operations are isomorphic via $\pi_i$, where $\pi_i (l) = \pi(l)$ if $l \in \{l_1,l_2\}$, $\pi_i (l_3)=l'_3$.
It is clear that, $G_g \cong^{\pi_i}_f H_h$ and $\pi_i/color = \pi|_{color(Q_g)}$.
 \par If there are only \emph{partial components}\footnoteref{note2} (see Section 5.2.1 in \cite{corneil2012polynomial}) in graph $G$ then 
the algorithm relabels all the vertices $V_{l_1}^a$ with $l_3$ and removes the edges $E_{l_3l_2}$ from $G$ to get the decomposition  with 
 top operations $\rho_{l_3 \rightarrow l_1}$ and $\eta_{l_3,l_2}$ 
 above a $\oplus$ operation with the connected components $G_{g_1}, \cdots ,G_{g_k}$ as children ($G_g = G_{g_1}\oplus \cdots \oplus G_{g_k}$).
 If $G \cong_f^{\pi} H$, the algorithm relabels all the vertices $V_{l'_1}^a$ in $H$ with $l'_3$ and removes the edges $E_{l'_3l'_2}$ to get the decomposition with top operations $\rho_{l'_3 \rightarrow l'_1}$ and $\eta_{l'_3,l'_2}$ 
above a $\oplus$ operation with the connected components $H_{h_1}, \cdots ,H_{h_k}$ as children (these are images of 
$G_{g_1}, \cdots ,G_{g_k}$ under $f$ in some order).
 The quotient graphs $Q_g$ and $Q_h$ build from top 
operations are isomorphic via $\pi_i$, where $\pi_i (l) = \pi(l)$ if $l \in \{l_1,l_2\}$, $\pi_i (l_3)=l'_3$,
and $G_g \cong^{\pi_i}_f H_h$, $\pi_i/color = \pi|_{color(Q_g)}$.
\par Lemma 30, 31 in \cite{corneil2012polynomial} shows that if $G \in \mathcal{U}_{l_1}$ apart from above two ways there is no other way to continue 
to find the decomposition for graphs of clique-width at most three. \qed
\end{proof}
\subsection{Function Decompose $\mathcal{\overline{D}}_{l_1}$:}
Let $V_l$ be the set of vertices with label $l$. 
For $l \in \{l_1,l_2\}$ we say, $G \in \mathcal{\overline{D}}_{l}$ if $\overline{V_{l}}$ is not connected and removing edges between the 
coconnected components of $V_{l}$ disconnects $G$. 
\begin{algorithm}
{\scriptsize 
\DontPrintSemicolon
\caption{Function $Decompose$-$leaf$-$\mathcal{\overline{D}}_{l_1}$ \cite{corneil2012polynomial}}
\label{D}
\KwIn{A bilabeled, $l$-prime and connected graph $G$}
\KwOut{true with top operations of parse tree or false if $cwd(G)$ $>$ 3}
\Begin{
tree := null\;
\If{$G$ has only two coconnected components }
{
Let $G_g =\oplus^k_{i=1}G_{g_i}$, where $G_{g_i}$'s are connected components of $G \setminus \{E_{l_3l_1}\}$\; 
$tree= \rho_{l_3 \rightarrow l_1}\eta_{l_3,l_1}(\oplus^k_{i=1}G_{g_i})$\;
return (true, tree)\;
}
\If { $G$ has proper partition }
{
Let $G_g =G_{g_1} \oplus G_{g_2}$, where $G_{g_1}$ and $G_{g_2}$ are connected components of $G \setminus \{E_{l_3l_1}\}$\; 
$tree= \rho_{l_3 \rightarrow l_1}\eta_{l_3,l_1}(G_{g_1}\oplus G_{g_2})$\;
return (true, tree)\;
}
\If { $G$ is eligible}
{
Let $G_g =y\oplus (G_{g_1}\setminus y) \oplus^k_{i=2}G_{g_i}$, where $G_{g_i}$'s are connected components of 
$G \setminus \{E_{l_3l_1}\}$ and $y, G_{g_1}\setminus y$'s are connected components of $G_{g_1} \setminus \{E_{l_3l_2}\}$ \; 
$tree= \rho_{l_3 \rightarrow l_1}\eta_{l_3,l_1}(\eta_{l_3l_2}(y\oplus G_{g_1}\setminus y)\oplus^k_{i=2}G_{g_i})$\; 
return (true, tree)\;
}
return(false,tree) (i.e., $cwd(G)$ $>$ 3)
}
}
\end{algorithm}

\begin{lemma}\label{lD1}
Let $G$ and $H$ be bilabeled, $l$-prime and connected graphs. If $G \cong _f^{\pi} H$ for some $f$ and $\pi$ then Algorithm~\ref{D} 
 generates top operations of parse trees $G$ and $H$ such that there is a $\pi_i$ $\in$ $\ISO(Q_g,Q_h)$ with  
$G_g \cong^{\pi_i}_f H_h$ and $\pi_i/color = \pi|_{color(Q_g)}$,  
where $G_g$ and $H_h$ are the graphs described in Algorithm~\ref{D}.
\end{lemma}
\begin{proof}
 The proof is divided into three cases depending on the structure of the graph.
 \par If there are only two coconnected components $CCC_1$ and $CCC_2$ of $V_{l_1}$, then the 
algorithm relabels one of $CCC_1$ or $CCC_2$ at random width $l_3$ and removes the edges $E_{l_3l_1}$ to get the 
 decomposition  with top operations $\rho_{l_3 \rightarrow l_1}$ and $\eta_{l_3,l_1}$ 
above a $\oplus$ operation with the connected components $G_{g_1}, \cdots ,G_{g_k}$ as children ($G_g = G_{g_1}\oplus \cdots \oplus G_{g_k}$).
 If $G \cong_f^{\pi} H$, up to a permutation of 
 labels $H$ may be in $\mathcal{\overline{D}}_{l'_1}$ or  $(\mathcal{\overline{D}}_{l'_2})$, without loss of generality assume $H$ is in $\mathcal{\overline{D}}_{l'_1}$.
 In $H$ the algorithm relabels one of the coconnected component of $V_{l'_1}$ at random with $l'_3$ and removes the edges $E_{l'_3l'_1}$ to get the 
 decomposition with top operations $\rho_{l'_3 \rightarrow l'_1}$ and $\eta_{l'_3,l'_1}$ above a $\oplus$ operation with the 
connected components $H_{h_1}, \cdots ,H_{h_k}$ as children (these are images of 
$G_{g_1}, \cdots ,G_{g_k}$ under $f$ in some order).
The quotient graphs $Q_g$ and $Q_h$ build from top operations are isomorphic via $\pi_i$, where $\pi_i (l_2) = l'_2$,  
$\pi_i (l_1)=l'_1$ or $l'_3$, $\pi_i (l_3)=l'_3$ or $l'_1$,  and $G_g \cong^{\pi_i}_f H_h$, $\pi_i/color = \pi|_{color(Q_g)}$. 

 \par If $G$ has a \emph {proper partition}\footnote{ partition of the coconnected components of $V_{l_1}$ into two sides such that the vertices of $V_{l_2}$ also 
partitions into two sides but no connected component of $V_{l_2}$ has vertices in both sides} then algorithm relabels one side of $V_{l_1}$ with $l_3$ and 
removes the edges $E_{l_3l_1}$ to get the decomposition  with top operations $\rho_{l_3 \rightarrow l_1}$ and $\eta_{l_3,l_1}$ 
above a $\oplus$ operation with the connected components $G_{g_1}$ and $G_{g_2}$ ($G_g = G_{g_1}\oplus G_{g_2}$) as children.
 If $G \cong_f^{\pi} H$, the algorithm relabels one side of $V_{l'_1}$  with $l'_3$ and removes the edges $E_{l'_3 l'_1}$ to get the decomposition with top operations $\rho_{l'_3 \rightarrow l'_1}$ and $\eta_{l'_3,l'_1}$ above a $\oplus$ operation with the
connected components $H_{h_1}$ and $H_{h_2}$ as children.
The quotient graphs $Q_g$ and $Q_h$ build from top operations are isomorphic via $\pi_i$, where $\pi_i (l_2) = l'_2$,  
$\pi_i (l_1)=l'_1$ or $l'_3$, $\pi_i (l_3)=l'_3$ or $l'_1$,  and $G_g \cong^{\pi_i}_f H_h$, $\pi_i/color = \pi|_{color(Q_g)}$. 
\par If $G$ is \emph{eligible} (see Section 5.2.2 in \cite{corneil2012polynomial} for definition) then 
to decompose $G$ the algorithm relabels vertices in coconnected component $CCC_d$ with $l_3$ and removes the edges $E_{l_3l_1}$ to get the decomposition with
top operations $\rho_{l_3\rightarrow l_1}$, $\eta_{l_3,l_1}$ and a $\oplus$
with the connected components $G_{g_1}, \cdots ,G_{g_k}$ as children ($G_g = G_{g_1}\oplus \cdots \oplus G_{g_k}$). Again the algorithm 
removes the edges $E_{l_3l_2}$ from $G_{g_1}$ to get the decomposition with top operations $\eta_{l_3,l_2}$ and a $\oplus$ 
with the connected components of $G_{g_1}=G_{g_1}\setminus \{y\}\oplus y$ as children.
If $G \cong_f^{\pi} H$, the algorithm relabels coconnected component $CCC'_d$ with $l'_3$ and  
removes the edges $E_{l'_3l'_1}$ from $H$ to get the decomposition with top operations  
are $\rho_{l'_3\rightarrow l'_1}$, $\eta_{l'_3,l'_1}$ and a $\oplus$
with the connected components $H_{h_1}, \cdots ,H_{h_k}$ (these are images of 
$G_{g_1}, \cdots ,G_{g_k}$ under $f$ in some order) as children. Again the algorithm
removes the edges $E_{l'_3,l'_2}$ from $H_{h_1}$ to get the decomposition with
top operations $\eta_{l'_3,l'_2}$ and a $\oplus$ with the connected components of $H_{h_1}=H_{h_1}\setminus \{y'\}\oplus y'$ as children.
In this case the generated parse tree has two levels.
In the first level the quotient graphs $Q_g$ and $Q_h$ built from top 
operations are isomorphic via $\pi_{1i}$, where $\pi_{1i} (l) = \pi(l)$ if $l \in \{l_1,l_2\}$, $\pi_{1i}(l_3)=l'_3$, 
and $G_g \cong^{\pi_{1i}}_f H_h$, $\pi_{1i}/color = \pi|_{color(Q_g)}$.
In the second level the quotient graphs $Q_{g_1}$ and $Q_{h_1}$ built from top
operations are isomorphic via $\pi_{2i}$, where $\pi_{2i}(l) = \pi_{1i}(l)$ if $l \in \{l_1,l_2,l_3\}$,  
and $G_{g_1} \cong^{\pi_{2i}}_f H_{h_1}$, $\pi_{2i}/color = \pi_{1i}|_{color(Q_{g_1})}$. 
\par Lemma 32 in \cite{corneil2012polynomial} shows that if $G \in \mathcal{\overline{D}}_{l_1}$ apart from the above three ways there is 
no other way to continue to find the decomposition for graphs of clique-width at most three. \qed
\end{proof}

\end{document}